\documentclass[a4paper]{article}%
\usepackage{amsfonts}
\usepackage{amsmath}
\usepackage{amssymb}
\usepackage{amstext}
\usepackage{graphicx}
\usepackage[a4paper]{geometry}
\usepackage[toc,page,title,titletoc,header]{appendix}%
\providecommand{\U}[1]{\protect\rule{.1in}{.1in}}
\newtheorem{theorem}{Theorem}[section]

\newtheorem{definition}[theorem]{Definition}

\newtheorem{lemma}[theorem]{Lemma}

\newtheorem{proposition}[theorem]{Proposition}
\newtheorem{remark}[theorem]{Remark}

\newenvironment{proof}[1][Proof]{\noindent\textbf{#1.} }{\ \rule{0.5em}{0.5em}}
\numberwithin{equation}{section}
\ifx\pdfoutput\relax\let\pdfoutput=\undefined\fi
\newcount\msipdfoutput
\ifx\pdfoutput\undefined\else
\ifcase\pdfoutput\else
\ifx\paperwidth\undefined\else
\ifdim\paperheight=0pt\relax\else\pdfpageheight\paperheight\fi
\ifdim\paperwidth=0pt\relax\else\pdfpagewidth\paperwidth\fi
\fi\fi\fi
\begin{document}

\title{Non-autonomous quantum systems with scale-dependent interface conditions.}
\author{Andrea Mantile\thanks{Laboratoire de Math\'{e}matiques, Universit\'{e} de
Reims - FR3399 CNRS, Moulin de la Housse BP 1039, 51687 Reims, France.}}
\date{}
\maketitle

\begin{abstract}
We consider a class of modified Schr\"{o}dinger operators where the
semiclassical Laplacian is perturbed with $h$-dependent interface conditions
occurring at the boundaries of the potential's support. Under positivity
assumptions on the potential, we show that this modification produces a small
perturbation on the dynamics as $h\rightarrow0$, independently from the time
scale. In the case of a time dependent potential, this yields uniform-in-$h$
stability estimates for products of instantaneous propagators. Then, following
a standard approach, the non-autonomous dynamical system is defined as a limit
of stepwise propagators and its small-$h$ expansion is provided under suitable
regularity assumptions on the potential's variations.

\end{abstract}

\section{Introduction}

Singular perturbations of the 1D Laplacian through non-mixed interface
conditions have been used as a technical tool for the study of the adiabatic
evolution of shape resonances in the asymptotic regime of quantum wells. A
possible approach to this problem consists in using a complex dilation to
identify the resonances and write the adiabatic problem as an evolution
equation for proper eigenstates of the deformed Hamiltonian with spectral gap
condition. This scheme does not allow a uniform-in-time estimate of the
resulting dynamical system and prevents a rigorous estimate of the error in
the adiabatic limit. An alternative approach developed in \cite{FMN2} consists
in modifying the physical Hamiltonian by replacing its kinetic part with a
perturbed Laplacian $\Delta_{\theta}$, whose domain is the restriction%
\begin{equation}
D(\Delta_{\theta})=H^{2}(\mathbb{R}\backslash\left\{  a,b\right\}
)\upharpoonright u:\left\{
\begin{array}
[c]{l}%
\smallskip e^{-\frac{\theta}{2}}u(b^{+})=u(b^{-})\,,\quad e^{-\frac{3}%
{2}\theta}u^{\prime}(b^{+})=u^{\prime}(b^{-})\,,\\
\\
e^{-\frac{\theta}{2}}u(a^{-})=u(a^{+}),\quad e^{-\frac{3}{2}\theta}u^{\prime
}(a^{-})=u^{\prime}(a^{+})\,,
\end{array}
\right.  \label{BC_theta}%
\end{equation}
(here: $\theta\in\mathbb{C}$ and $u(x^{\pm})$ denote the right and left limits
of the function in $x$), while the action is given by: $\Delta_{\theta
}u(x)=u^{\prime\prime}(x)$ for $x\in\mathbb{R}\backslash\left\{  a,b\right\}
$. The corresponding modified Hamiltonian: $H_{\theta}^{h}=-h^{2}%
\Delta_{\theta}+\mathcal{V}\,$, is defined with a potential $\mathcal{V}$
compactly supported on $\left[  a,b\right]  $ and possibly depending on $h$.
The potential's profile can be chosen in order to fix some required spectral
condition (as the existence of shape resonances): in this framework, $h$ is a
small parameter fixing the quantum scale of the model.

According to the analytic dilation technique, the resonances of $H_{\theta
}^{h}$ identify, in the sector \newline$\left\{  z\in\mathbb{C}%
\,,\ -2\operatorname{Im}\varphi<\arg z\leq0\right\}  $, with the spectral
points of the deformed operator: \newline$\left.  H_{\theta}^{h}\left(
\varphi\right)  =-h^{2}e^{-2\varphi\,1_{\mathbb{R}\backslash\left(
a,b\right)  }(x)}\Delta_{\theta+\varphi}+\mathcal{V}\right.  $ resulting from
the sharp exterior dilation: $x\rightarrow e^{\theta1_{\mathbb{R}%
\backslash\left(  a,b\right)  }(x)}x$ (see \cite[Proposition 3.6]{FMN2}; see
also \cite[Chp. 16]{HiSig} and the references therein for an introduction to
the complex deformation method). The interest in these models stands upon the
fact that, when $\theta=\varphi$ and $\operatorname{Im}\theta>0$, the
perturbed Laplacian $-i\Delta_{\theta}$ transforms into the maximal accretive
operator: $\left.  -ie^{-2\theta\,1_{\mathbb{R}\backslash\left(  a,b\right)
}(x)}\Delta_{2\theta}\right.  $ (being $1_{\mathbb{R}\backslash\left(
a,b\right)  }$ the characteristic function of the exterior domain). Hence, the
corresponding complex dilated Hamiltonian: $H_{\theta}^{h}\left(
\theta\right)  $, although nonselfadjoint, is the generator of a quantum
semigroup of contractions (we refer to the Lemma 3.1 in \cite{FMN2}). In the
time-dependent case, this allows to rephrase the adiabatic evolution problem
for the resonances of $H_{\theta}^{h}\left(  t\right)  =-h^{2}\Delta_{\theta
}+\mathcal{V}\left(  t\right)  $ as an adiabatic problem for the corresponding
eigenstates of $H_{\theta}^{h}\left(  \theta,t\right)  $ and, accounting the
contractivity property of $e^{-itH_{\theta}^{h}\left(  \theta,t\right)  }$, a
'standard' adiabatic theory can be developed (e.g. in \cite{Nenciu}). This
approach leaded to a version of the adiabatic theorem holding for shape
resonances in the regime of quantum wells in a semiclassical island
\cite[Theorem 7.1]{FMN2}.

The error introduced using modified Hamiltonians of the type $H_{\theta}^{h}$
is determined by the difference between the modified dynamics and the unitary
evolution generated by the corresponding selfadjoint operator $H_{0}^{h}$. To
justify the use of $H_{\theta}^{h}$ in modelling realistic physical
situations, this error needs to be carefully estimated uniformly-in-time when
$\theta$ is assumed to be small. The case of time independent potentials have
been considered in \cite{Man1}-\cite{Man2}. It is worthwhile to mention at
this concern that, for $\theta$ small, $H_{\theta}^{h}$ is neither selfadjoint
nor symmetric. Hence the definition of the quantum dynamics generated by the
modified operator does not follows using standard arguments from selfadjoint
theory. For $h=1$, an accurate resolvent analysis, and explicit formulas for
the generalized eigenfunctions of the modified operator, allow to obtain a
small-$\theta$ expansion of the stationary waves operators for couple
$\left\{  H_{\theta}^{1},H_{0}^{1}\right\}  $, provided that $\mathcal{V}\in
L^{2}\left(  \mathbb{R}\right)  $ is compactly supported on $\left[
a,b\right]  $ and defined positive. This yields an uniform-in-time estimate
for the 'distance' between the two dynamics according to the expansion:
$e^{-itH_{\theta}^{1}}=e^{-itH_{0}^{1}}+\mathcal{O}\left(  \left\vert
\theta\right\vert \right)  $, where $\mathcal{O}\left(  \cdot\right)  $ is
intended in the $L^{2}$-operator norm sense (see theorem 1.2 in \cite{Man1}).
The case of $h$-dependent models is considered in \cite{Man2} under the
asymptotic regime of quantum wells in a semiclassical island; this particular
framework, realized with a potential formed by the superposition of a
potential barrier supported on $\left[  a,b\right]  $ and potential wells
supported in $\left(  a,b\right)  $ with support of size $h\sim0$, has been
shown to be relevant for the modelling of quantum transport systems where the
charge careers couples with quasi-stationary quantum states (see e.g. in
\cite{BNP1}, \cite{JoPrSj}, \cite{Capa}). Using specific spectral assumptions,
a small-$\theta$ expansion for the modified dynamical system $e^{-itH_{\theta
}^{h}}$, similar to the one given in \cite{Man1}, has been obtained in this
case (see Theorem 4.4 in \cite{Man2}); nevertheless, when both $h$ and
$\theta$ are small, the resulting error term is small uniformly in time only
for initial states belonging to an appropriate subspace with prescribed energy
conditions. This prevents to extend the result to the most relevant case of
time dependent potentials.

In this work we first reconsider the case of autonomous potentials under
generic assumptions. In particular, for $\mathcal{V}\in L^{\infty}\left(
\mathbb{R}\right)  $ with compact support on $\left[  a,b\right]  $ and
$1_{\left[  a,b\right]  }\mathcal{V}>0$, we provide a small-$\theta$ expansion
of the propagator $e^{-itH_{\theta}^{h}}$ globally holding on $L^{2}\left(
\mathbb{R}\right)  $ uniformly w.r.t. $t\in\mathbb{R}$ and $h\in\left(
0,h_{0}\right)  $. This generalizes the result obtained in \cite{Man2} and
allows us to discuss the case of time dependent potentials by adapting the
Kato-Yoshida construction of the modified dynamics in terms of piecewise
product of modified propagators. Although not strictly tailored on the case of
the quantum well regime, our assumptions also includes the case of
$h$-dependent potentials, while the positivity constraint on $\mathcal{V}$ is
still coherent with the description of a potential island generating shape
resonances; in this connection, the operators concerned with our work can be
adapted to the modelling of physical systems involving the interaction with
quasi-stationary states corresponding to shape resonances.

\section{\label{Sec_Model}Models and results}

We consider the modified Schr\"{o}dinger operators%
\begin{equation}
D\left(  H_{\theta}^{h}\right)  =D\left(  \Delta_{\theta}\right)
\,,\qquad\left(  H_{\theta}^{h}\,u\right)  (x)=-h^{2}u^{\prime\prime
}(x)+\mathcal{V}(x)\,u(x)\,,\qquad x\in\mathbb{R}\backslash\left\{
a,b\right\}  \,, \label{H_teta_h}%
\end{equation}
where $\Delta_{\theta}$ is defined according to (\ref{BC_theta}). The domain
$D\left(  \Delta_{\theta}\right)  $ is next considered as an Hilbert subspace
of $H^{2}\left(  \mathbb{R}\backslash\left\{  a,b\right\}  \right)  $ (see the
definition (\ref{BC_theta})).

\begin{theorem}
\label{Theorem_propagator}Let $h\in\left(  0,h_{0}\right]  $, $c>0$,
$\left\vert \theta\right\vert \leq h^{N_{0}}$ with $h_{0}$ suitably small and
$N_{0}>2$. For $\mathcal{V}$ defined according to%
\begin{equation}%
\begin{array}
[c]{ccccc}%
\mathcal{V}\in L^{\infty}\left(  \mathbb{R}\right)  \,, &  &
\text{supp\thinspace}\mathcal{V}=\left[  a,b\right]  \,, &  & 1_{\left[
a,b\right]  }\mathcal{V}>c\,,
\end{array}
\label{V_gen}%
\end{equation}
The map $iH_{\theta}^{h}$ generates a strongly continuous group of bounded
operators both on $L^{2}(\mathbb{R})$ and on the Hilbert space $D\left(
\Delta_{\theta}\right)  $ equipped with the $H^{2}\left(  \mathbb{R}%
\backslash\left\{  a,b\right\}  \right)  $-norm. For $u\in D\left(
\Delta_{\theta}\right)  $ the identity: $i\partial_{t}\left(  e^{-itH_{\theta
}^{h}}u\right)  =H_{\theta}^{h}e^{-itH_{\theta}^{h}}u$, holds in the
$L^{2}(\mathbb{R})$-sense.

For a fixed $t$, $e^{-itH_{\theta}^{h}}$ is $\theta$-analytic w.r.t. the
$L^{2}(\mathbb{R})$-operator norm and allows the expansion%
\begin{equation}
\sup\limits_{t\in\mathbb{R\,}}\left\Vert e^{-itH_{\theta}^{h}}-e^{-itH_{0}%
^{h}}\right\Vert _{\mathcal{L}\left(  L^{2}(\mathbb{R})\right)  }\leq
C_{a,b,c}\,h^{N_{0}-2}\,, \label{propagator_est_1}%
\end{equation}
with $C_{a,b,c}>0$, possibly depending on the data $a,b,c$.
\end{theorem}

The case of a time dependent Hamiltonian%
\begin{equation}
H_{\theta}^{h}\left(  t\right)  =-h^{2}\Delta_{\theta}+\mathcal{V}\left(
t\right)  \label{H_teta_h_t}%
\end{equation}
is analyzed under the assumptions%
\begin{equation}
\left.
\begin{array}
[c]{ccccc}%
\mathcal{V}\left(  t\right)  \in\mathcal{C}^{0}\left(  \left[  0,T\right]
,L^{\infty}(\mathbb{R},\mathbb{R})\right)  \,, &  & \text{supp\thinspace
}\mathcal{V}\left(  t\right)  =\left[  a,b\right]  \,, &  & 1_{\left[
a,b\right]  }\mathcal{V}\left(  t\right)  >c\,,
\end{array}
\right.  \label{V_cond_1}%
\end{equation}
for some $c>0$, and%
\begin{equation}%
\begin{array}
[c]{ccc}%
\mathcal{V}\left(  t\right)  -\mathcal{V}\left(  s\right)  \in W_{0}%
^{2,\infty}\left(  \left[  a,b\right]  \right)  \,, &  & \forall
\,t,s\in\left[  T,0\right]  \,.
\end{array}
\label{V_cond_2}%
\end{equation}
where%
\begin{equation}
W_{0}^{2,\infty}\left(  \left[  a,b\right]  \right)  =\left\{  \psi\in
W^{2,\infty}\left(  \left[  a,b\right]  \right)  \,\left\vert \ \psi\left(
\alpha\right)  =\psi^{\prime}\left(  \alpha\right)  =0\,,\alpha=a,b\right.
\right\}  \,. \label{Multipliers}%
\end{equation}
The small-$\theta$ behaviour of the resulting quantum dynamical system is
characterized as follows.

\begin{theorem}
\label{Theorem_main}Let $h\in\left(  0,h_{0}\right]  $, $\left\vert
\theta\right\vert \leq h^{N_{0}}$ with $h_{0}$ suitably small and $N_{0}>2$,
and assume $H_{\theta}^{h}\left(  t\right)  $ to be defined according to
(\ref{H_teta_h_t}). Under the conditions (\ref{V_cond_1})-(\ref{V_cond_2}),
there exists a unique family of operators $U_{\theta}^{h}\left(  t,s\right)
$, bounded and strongly continuous in $t$ and $s$ w.r.t. the $L^{2}\left(
\mathbb{R}\right)  $-operator norm, fulfilling the identities%
\begin{equation}
U_{\theta}^{h}\left(  s,s\right)  =1_{L^{2}\left(  \mathbb{R}\right)
}\,,\quad U_{\theta}^{h}\left(  t,s\right)  =U_{\theta}^{h}\left(  t,r\right)
U_{\theta,n}^{h}\left(  r,s\right)  \,,\quad\forall\,s\leq r\leq t\,,
\label{propagator_id_lim}%
\end{equation}
and such that $U_{\theta}^{h}\left(  t,s\right)  u$ is the solution of the
Cauchy problem%
\begin{equation}
\left\{
\begin{array}
[c]{l}%
i\partial_{t}u_{\theta}^{h}\left(  t\right)  =H_{\theta}^{h}\left(  t\right)
u_{\theta}^{h}\left(  t\right)  \,,\quad0\leq s\leq t\leq T\,,\\
\\
u_{\theta}^{h}\left(  s\right)  =u\in D\left(  \Delta_{\theta}\right)  \,.
\end{array}
\right.  \label{Cauchy_pb}%
\end{equation}
Moreover, $U_{\theta}^{h}\left(  t,s\right)  $ is $\theta$-holomorphic in
$\mathcal{L}\left(  L^{2}\left(  \mathbb{R}\right)  \right)  $ and allows the
estimates%
\begin{equation}
\sup_{s,t\in\left[  0,T\right]  }\left\Vert U_{\theta}^{h}\left(  t,s\right)
-U_{0}^{h}\left(  t,s\right)  \right\Vert _{\mathcal{L}\left(  L^{2}\left(
\mathbb{R}\right)  \right)  }\leq M_{a,b,c}\sup_{t\in\left[  0,T\right]
}\left\Vert \mathcal{V}\left(  t\right)  \right\Vert _{L^{\infty}\left(
\mathbb{R}\right)  }\,h^{N_{0}-2-\delta}\,, \label{propagator_est_global_h}%
\end{equation}
with $M_{a,b,c}>0$ depending on the data, but independent from $T$, and
$\delta>0$ arbitrarily small.
\end{theorem}

Our modified model $H_{\theta}^{h}$ identifies with an extension of the
symmetric operator%
\begin{equation}
H_{0,0}^{h}=H_{0}^{h}\upharpoonright\left\{  u\in H^{2}\left(  \mathbb{R}%
\right)  \,\left\vert \ u(\alpha)=u^{\prime}(\alpha)=0\,,\ \alpha=a,b\right.
\right\}  \,. \label{H_00}%
\end{equation}
Hence, $H_{\theta}^{h}$ is explicitly solvable w.r.t. $H_{0}^{h}$ and relevant
quantities, as its resolvent or generalized eigenfunctions, can be expressed
in terms of corresponding non-modified quantities, related to $H_{0}^{h}$,
through non-perturbative formulas. This well-known property of point
perturbations (see e.g. in \cite{Albe}) provides with an useful tool for the
spectral analysis and allows to consider the pair $\left\{  H_{\theta}%
^{h},H_{0}^{h}\right\}  $ as a scattering system. Following this approach, in
the Section \ref{Sec_Modop}, we use a 'Krein-like' formula for the modified
generalized eigenfunctions: this yields a non-perturbative representation of
the stationary wave operators related to the couple $\left\{  H_{\theta}%
^{h},H_{0}^{h}\right\}  $. Thus, in the autonomous case, the quantum dynamics
generated by $H_{\theta}^{h}$ is determined by conjugation from $e^{-itH_{0}%
^{h}}$ and an uniform-in-time estimate for the 'distance' between the two
dynamics follows from the small-$\theta$ behaviour of the stationary waves
operators. In particular, the positivity assumption for the potential allows
to extend the result obtained in \cite{Man2}, providing with the
small-$\theta$ expansion (\ref{propagator_est_1}) for the modified quantum
propagator holding without restrictions on the initial state.

In the non-autonomous case, considered in the Section \ref{Sec_Nonaut}, the
dynamics is approximated by a stepwise product of propagators associated to
the 'instantaneous' Hamiltonians. The stability of the approximating dynamics
w.r.t. the $\mathcal{L}\left(  L^{2}\left(  \mathbb{R}\right)  \right)  $ and
$\mathcal{L}\left(  D\left(  \Delta_{\theta}\right)  \right)  $ topologies is
discussed in the Lemma \ref{Lemma_propagator_n}; using this result, its
uniform convergence is obtained in the Proposition
\ref{Proposition_propagator_lim} by adapting the approach of \cite{Kato1} to
our nonselfadjoint framework.

The expansion (\ref{propagator_est_global_h}), obtained under the conditions
(\ref{V_cond_1})-(\ref{V_cond_2}), describes the asymptotic behaviour of the
modified dynamical system, as $h,\theta\rightarrow0$, in the case of
time-dependent potentials. It is worthwhile to remark that the prescription:
$\left.  \mathcal{V}\left(  t\right)  -\mathcal{V}\left(  s\right)  \in
W_{0}^{2,\infty}\left(  \left[  a,b\right]  \right)  \right.  $ in
(\ref{V_cond_2}) is coherent with the modelling of quantum transport systems
where the variations in time of the potential, determined by the (possibly
nonlinear) time evolution of quasi-stationary states, are expected to be
concentrated in small regions of the device (i.e. inside $\left(  a,b\right)
$). Moreover, this result does not depends on the time scale $T$: hence, it
can be possibly adapted to the analysis of adiabatic evolution problems, where
the potential's variation rate is fixed by $\varepsilon$ small, and the
natural time scale grows according to $1/\varepsilon$.

\subsection{Notation}

In what follows: $B_{\delta}(p)$ is the open disk of radius $\delta$ centered
in a point $p\in\mathbb{C}$; $\mathbb{C}^{+}$ is the upper complex half-plane;
$1_{\Omega}(\cdot)$ is the characteristic function of a domain $\Omega$;
$\partial_{j}f$, denotes the derivative of $f$ w.r.t. the $j$-th variable;
$\mathcal{C}_{x}^{n}(U)$ is the set of $\mathcal{C}^{n}$-continuous functions
w.r.t. $x\in U\subseteq\mathbb{R}$, while $\mathcal{H}_{z}(D)$ is the set of
holomorphic functions w.r.t. $z\in D\subseteq\mathbb{C}$. The notation
'$\lesssim$', appearing in some of the proofs, denotes the inequality: '$\leq
C$' being $C$ a suitable positive constant. Moreover, the generalization of
the Landau notation $\mathcal{O}\left(  \cdot\right)  $ is defined according to

\begin{definition}
\label{Landau_Notation}Let be $X$ a metric space and $f,g:X\rightarrow
\mathbb{C}$. Then $f=\mathcal{O}\left(  g\right)  $
$\overset{def}{\Longleftrightarrow}$ $\forall\,x\in X$ it holds: $\left.
f(x)=p(x)g(x)\right.  $, being $p$ a bounded map $X\rightarrow\mathbb{C}$.
\end{definition}

\section{\label{Sec_Modop}Scattering by interface conditions}

Following the analysis developed in \cite{Man1},\cite{Man2}, we next resume
the main features of the operators $H_{\theta}^{h}$, focusing on the
scattering couple $\left\{  H_{\theta}^{h},H_{0}^{h}\right\}  $.

In the case $h=1$, it has been shown that the interface conditions
(\ref{BC_theta}) do not modify the spectrum provided that $\theta$ is small
(see \cite[Proposition 2.6]{Man1}). In the present case, the dilation:
$y=\left(  x-\left(  b+a\right)  /2\right)  /h$ transforms the boundary
conditions (\ref{BC_theta}) into%
\begin{equation}
\left\{
\begin{array}
[c]{l}%
\smallskip e^{-\theta/2}u(\beta_{h}^{+})=u(\beta_{h}^{-})\,,\quad
e^{-\theta\,3/2}u^{\prime}(\beta_{h}^{+})=u^{\prime}(\beta_{h}^{-})\,,\\
\\
e^{-\theta/2}u(\alpha_{h}^{-})=u(\alpha_{h}^{+}),\quad e^{-\theta
\,3/2}u^{\prime}(\alpha_{h}^{-})=u^{\prime}(\alpha_{h}^{+})\,,
\end{array}
\right.  \label{B_C_h}%
\end{equation}
with: $\alpha_{h}=-\left(  b-a\right)  /2h$ and $\beta_{h}=\left(  b-a\right)
/2h$. The corresponding unitary map on $L^{2}\left(  \mathbb{R}\right)  $
transforms $H_{\theta}^{h}$ into the dilated operator%
\begin{equation}
\tilde{H}_{\theta}:\left\{
\begin{array}
[c]{l}%
D\left(  \tilde{H}_{\theta}\right)  =\left\{  u\in H^{2}\left(  \mathbb{R}%
\backslash\left\{  \alpha_{h},\beta_{h}\right\}  \right)  \,\left\vert
\text{(\emph{\ref{B_C_h}}) holds}\right.  \right\}  \,,\\
\\
\left(  \tilde{H}_{\theta}\,u\right)  (x)=-u^{\prime\prime}(x)+\mathcal{\tilde
{V}}(x)\,u(x)\,,\qquad x\in\mathbb{R}\backslash\left\{  \alpha_{h},\beta
_{h}\right\}  \,,
\end{array}
\right.  \label{H_teta_dil}%
\end{equation}
where $\mathcal{\tilde{V}}\left(  x\right)  =\mathcal{V}\left(  hx+\left(
b+a\right)  /2\right)  $ is compactly supported on $\left[  \alpha_{h}%
,\beta_{h}\right]  $.

\begin{proposition}
\label{Proposition_spectrum}Let $h>0$ fixed and consider the operators
$H_{\theta}^{h}$ defined in (\ref{H_teta_h}) with%
\begin{equation}%
\begin{array}
[c]{ccc}%
\mathcal{V}\in L^{2}(\mathbb{R},\mathbb{R})\,, &  & \text{supp }%
\mathcal{V}=\left[  a,b\right]  \,.
\end{array}
\label{V}%
\end{equation}
For any couple $\theta\in\mathbb{C}$, the essential part of the spectrum is
$\sigma_{ess}\left(  H_{\theta}^{h}\right)  =\mathbb{R}_{+}$. If, in addition,
$\mathcal{V}$ is assumed to be defined positive, i.e.%
\begin{equation}
\left\langle u,\mathcal{V\,}u\right\rangle _{L^{2}(\left(  a,b\right)
)}>0\qquad\forall\,u\in L^{2}(\left(  a,b\right)  )\ s.t.\ u\neq0\,,
\label{V_pos}%
\end{equation}
it exists $\delta>0$, possibly depending on $h$, such that: $\sigma\left(
H_{\theta}^{h}\right)  =\mathbb{R}_{+}$ for all $\theta\in B_{\delta}\left(
0\right)  $.
\end{proposition}

\begin{proof}
From to the Proposition 2.6 in \cite{Man1}, the result holds for $\tilde
{H}_{\theta}$; then it extends to $H_{\theta}^{h}$ due to the unitarily
equivalence of the two operators.
\end{proof}

\begin{description}
\item[\emph{Notice}] The essential spectrum of $A$ is here defined according
to \cite{Wolf} as $\sigma_{ess}\left(  A\right)  =\mathbb{C}\backslash
\mathcal{F}\left(  A\right)  $, being $\mathcal{F}\left(  A\right)  $ the set
of complex $\lambda\in\mathbb{C}$ s.t. $\left(  A-\lambda\right)  $ is Fredholm.
\end{description}

The point perturbation model $H_{\theta}^{h}$ can be described as a
restriction of the adjoint operator $\left(  H_{0,0}^{h}\right)  ^{\ast}$ (see
\ref{H_00}) through linear relations on an auxiliary Hilbert space. This
construction, achieved in \cite{Man2} using the 'boundary triples' technique,
allows to express the difference $\left.  \left(  H_{\theta}^{h}-z\right)
^{-1}-\left(  H_{0}^{h}-z\right)  ^{-1}\right.  $ in terms of a finite rank
operator with range: $\ker\left(  \left(  H_{0,0}^{h}\right)  ^{\ast
}-z\right)  $. A basis of this defect space is formed by the Green's functions
associated to the differential operator $\left(  -h^{2}\partial_{x}%
^{2}+\mathcal{V-}z\right)  $, and by their first derivatives. Let $z\in
res\left(  H_{0}^{h}\right)  $ and introduce $\mathcal{G}^{z,h}(x,y)$ and
$\mathcal{H}^{z,h}(x,y)$ as solutions of the boundary value problems%
\begin{equation}
\left\{
\begin{array}
[c]{lll}%
\left(  -h^{2}\partial_{x}^{2}+\mathcal{V-}z\right)  \mathcal{G}^{z,h}%
(\cdot,y)=0\,, &  & \text{in }\mathbb{R}/\left\{  y\right\}  \,,\\
&  & \\
\mathcal{G}^{z,h}(y^{+},y)=\mathcal{G}^{z,h}(y^{-},y)\,, &  & h^{2}\left(
\partial_{1}\mathcal{G}^{z}(y^{+},y)-\partial_{1}\mathcal{G}^{z}%
(y^{-},y)\right)  =-1\,,
\end{array}
\right.  \label{Green_eq}%
\end{equation}
and%
\begin{equation}
\left\{
\begin{array}
[c]{lll}%
\left(  -h^{2}\partial_{x}^{2}+\mathcal{V-}z\right)  \mathcal{H}^{z,h}%
(\cdot,y)=0\,, &  & \text{in }\mathbb{R}/\left\{  y\right\}  \,,\\
&  & \\
h^{2}\left(  \mathcal{H}^{z,h}(y^{+},y)-\mathcal{H}^{z,h}(y^{-},y)\right)
=1\,, &  & \partial_{1}\mathcal{H}^{z}(y^{+},y)=\partial_{1}\mathcal{H}%
^{z}(y^{-},y)\,,
\end{array}
\right.  \label{D_Green_eq}%
\end{equation}
Then $\ker\left(  \left(  H_{0,0}^{h}\right)  ^{\ast}-z\right)  =l.c.\left\{
\gamma_{z,h,j}\right\}  _{j=1}^{4}$, with%
\begin{equation}
\gamma_{z,h,1}=\mathcal{G}^{z,h}(x,b)\,,\ \gamma_{z,h,2}=\mathcal{H}%
^{z,h}(x,b)\,,\ \gamma_{z,h,3}=\mathcal{G}^{z,h}(x,a)\,,\ \gamma
_{z,h,4}=\mathcal{H}^{z,h}(x,a)\,. \label{defect1}%
\end{equation}
Following \cite[eq. (2.19) and (2.26)]{Man2}, for all $z\in res\left(
H_{0}^{h}\right)  $ the identity
\begin{equation}
\left(  H_{\theta}^{h}-z\right)  ^{-1}-\left(  H_{0}^{h}-z\right)  ^{-1}%
=-\sum_{i,j=1}^{4}\left[  \left(  B_{\theta}\,q(z,h)-A_{\theta}\right)
^{-1}B_{\theta}\right]  _{ij}\left\langle \gamma_{\bar{z},h,j},\cdot
\right\rangle _{L^{2}(\mathbb{R})}\gamma_{z,h,j}\,, \label{krein_1}%
\end{equation}
holds with%
\begin{equation}%
\begin{array}
[c]{cc}%
h^{2}A_{\theta}-1=%
\begin{pmatrix}
e^{\theta3/2} &  &  & \\
& e^{\theta/2} &  & \\
&  & e^{-\theta3/2} & \\
&  &  & e^{-\theta/2}%
\end{pmatrix}
\,, & B_{\theta}=%
\begin{pmatrix}
0 & 1-e^{\theta3/2} &  & \\
e^{\theta/2}-1 & 0 &  & \\
&  & 0 & 1-e^{-\theta3/2}\\
&  & e^{-\theta/2}-1 & 0
\end{pmatrix}
\,,
\end{array}
\label{AB_theta_matrix}%
\end{equation}
$\,$and $q(z,h)$ depending on the boundary values of $\mathcal{G}^{z,h}$,
$\mathcal{H}^{z,h}$ and $\partial_{1}\mathcal{H}^{z,h}$ according to%
\begin{align}
&  q(z,h)\label{q_z}\\
&  =%
\begin{pmatrix}
\mathcal{G}^{z,h}(b,b)\medskip & -\left(  \mathcal{H}^{z,h}(b^{-},b)+\frac
{1}{2h^{2}}\right)  & \mathcal{G}^{z,h}(b,a) & -\mathcal{H}^{z,h}(b,a)\\
\mathcal{H}^{z,h}(b^{-},b)+\frac{1}{2h^{2}} & -\partial_{1}\mathcal{H}%
^{z,h}(b,b) & \mathcal{H}^{z,h}(a,b) & -\partial_{1}\mathcal{H}^{z,h}(b,a)\\
\mathcal{G}^{z,h}(a,b)\medskip & -\mathcal{H}^{z,h}(a,b) & \mathcal{G}%
^{z,h}(a,a) & -\left(  \mathcal{H}^{z,h}(a^{+},a)-\frac{1}{2h^{2}}\right) \\
\mathcal{H}^{z,h}(b,a) & -\partial_{1}\mathcal{H}^{z,h}(a,b) & \mathcal{H}%
^{z,h}(a^{+},a)-\frac{1}{2h^{2}} & -\partial_{1}\mathcal{H}^{z,h}(a,a)
\end{pmatrix}
\,.\nonumber
\end{align}

The Green's functions $\mathcal{G}^{z,h}$, $\mathcal{H}^{z,h}$ are related to
the Jost's solutions of the equation%
\begin{equation}%
\begin{array}
[c]{ccc}%
\left(  -h^{2}\partial_{x}^{2}+\mathcal{V}\right)  u=\zeta^{2}u\,, &  &
\zeta\in\mathbb{C}^{+}\,,
\end{array}
\label{Jost_eq_h1}%
\end{equation}
next denoted with $\chi_{\pm}^{h}\left(  \cdot,\zeta\right)  $, fulfilling the
exterior conditions%
\begin{equation}
\left.  \chi_{+}^{h}\left(  \cdot,\zeta\right)  \right\vert _{x>b}%
=e^{i\frac{\zeta}{h}x}\,,\qquad\left.  \chi_{-}^{h}\left(  \cdot,\zeta\right)
\right\vert _{x<a}=e^{-i\frac{\zeta}{h}x}\,. \label{Jost_sol_ext_h}%
\end{equation}
A detailed analysis of their properties have been given in \cite{Yafa} for
generic $L^{1}$-potentials, while the particular case of a potential barrier
is explicitly considered in \cite{Man1} for $h=1$. The $h$-dependent case can
be considered as a rescaled problem and the result presented in \cite{Man1}
rephrase as follows

\begin{lemma}
\label{Lemma_Jost}Let $\mathcal{V}\in L^{2}(\mathbb{R},\mathbb{R})$ s.t.: supp
$\mathcal{V}=\left[  a,b\right]  $. For any fixed $h>0$, the solutions
$\chi_{\pm}^{h}$ to the problem (\ref{Jost_eq_h1})-(\ref{Jost_sol_ext_h})
belong to $\mathcal{C}_{x}^{1}\left(  \mathbb{R},\,\mathcal{H}_{\zeta}\left(
\mathbb{C}^{+}\right)  \right)  $ and have continuous extension to the real axis.
\end{lemma}

\begin{proof}
With the change of variable: $y=\left(  x-\left(  b+a\right)  /2\right)  /h$,
the problem (\ref{Jost_eq_h1})-(\ref{Jost_sol_ext_h}) writes as%
\begin{equation}
\left\{
\begin{array}
[c]{lll}%
\left(  -\partial_{y}^{2}+\mathcal{\tilde{V}}\right)  \tilde{\chi}_{\pm}%
^{h}=\zeta^{2}\tilde{\chi}_{\pm}^{h}\,, &  & \zeta\in\mathbb{C}^{+}\,,\\
&  & \\
\left.  \tilde{\chi}_{+}^{h}\left(  \cdot,\zeta\right)  \right\vert
_{y>\left(  b-a\right)  /2h}=e^{i\zeta y}\,, &  & \left.  \tilde{\chi}_{-}%
^{h}\left(  \cdot,\zeta\right)  \right\vert _{y<-\left(  b-a\right)
/2h}=e^{-i\zeta y}\,,
\end{array}
\right.  \label{Jost_eq_h_rescaled}%
\end{equation}
where $\mathcal{\tilde{V}}$ denotes the dilated potential: $\mathcal{\tilde
{V}}\left(  y\right)  =\mathcal{V}\left(  hy+\left(  b+a\right)  /2\right)  $,
supported on $\left[  -\left(  b-a\right)  /2h,\left(  b-a\right)  /2h\right]
$, while $\tilde{\chi}_{\pm}^{h}$ correspond to the rescaled Jost's functions%
\begin{equation}
\tilde{\chi}_{\pm}^{h}\left(  y\right)  =\chi_{\pm}^{h}\left(  hy+\left(
b+a\right)  /2\right)  e^{-i\frac{\zeta}{2h}\left(  b+a\right)  }\,.
\label{Jost_fun_rescaled}%
\end{equation}
In this framework the Proposition 2.2 in \cite{Man1} applies; this yield
$\tilde{\chi}_{\pm}^{h}\in\mathcal{C}_{y}^{1}\left(  \mathbb{R},\,\mathcal{H}%
_{\zeta}\left(  \mathbb{C}^{+}\right)  \right)  $ with continuous extensions
to the closed complex half-plane $\overline{\mathbb{C}^{+}}$.
\end{proof}

Let $\zeta\in\mathbb{C}^{+}$ be such that: $\zeta^{2}\in res\left(  H_{0}%
^{h}\right)  $; rephrasing the relation \cite[Chp. 5, eq. (1.10)]{Yafa} in our
framework, we get%
\begin{align}
\mathcal{G}^{\zeta^{2},h}\left(  \cdot,y\right)   &  =\frac{1}{h^{2}%
w^{h}\left(  \zeta\right)  }\mathcal{\,}\left\{
\begin{array}
[c]{c}%
\chi_{+}^{h}\left(  \cdot,\zeta\right)  \chi_{-}^{h}\left(  y,\zeta\right)
\,,\qquad x\geq y\,,\\
\\
\chi_{-}^{h}\left(  \cdot,\zeta\right)  \chi_{+}^{h}\left(  y,\zeta\right)
\,,\qquad x<y\,,
\end{array}
\right. \label{G_z_h}\\
& \nonumber\\
\mathcal{H}^{\zeta^{2},h}\left(  \cdot,y\right)   &  =\frac{1}{h^{2}%
w^{h}\left(  \zeta\right)  }\mathcal{\,}\left\{
\begin{array}
[c]{c}%
\chi_{+}^{h}\left(  \cdot,\zeta\right)  \,\partial_{1}\chi_{-}^{h}\left(
y,\zeta\right)  \,,\qquad x\geq y\,,\\
\\
\chi_{-}^{h}\left(  \cdot,\zeta\right)  \,\partial_{1}\chi_{+}^{h}\left(
y,\zeta\right)  \,,\qquad x<y\,.
\end{array}
\right.  \label{H_z_h}%
\end{align}
where $w^{h}\left(  \zeta\right)  $, depending only on $\zeta$ and
$\mathcal{V}$, denotes the Wronskian associated to the couple $\left\{
\chi_{+}^{h}\left(  \cdot,\zeta\right)  ,\chi_{-}^{h}\left(  \cdot
,\zeta\right)  \right\}  $ (defined by: $w\left(  f,g\right)  =fg^{\prime
}-f^{\prime}g$). Due to the result of the Lemma \ref{Lemma_Jost}, for each
$h>0$, the maps $z\rightarrow\mathcal{G}^{z,h}(x,y)$, $z\rightarrow
\mathcal{H}^{z,h}(x,y)$ are meromorphic in $\mathbb{C}\backslash\mathbb{R}%
_{+}$ with a branch cut along the positive real axis and poles, corresponding
to the points in $\sigma_{p}\left(  H_{0}^{h}\right)  $ located on the
negative real axis. Adapting \cite[Chp. 5, eq. (1.9)]{Yafa} to the
$h$-dependent case, the function $w^{h}\left(  k\right)  $ fulfills the
identity: $\left\vert w^{h}(k)\right\vert ^{2}=k^{2}/h^{2}+\left\vert
w_{0}^{h}(k)\right\vert ^{2}$, where $w_{0}^{h}(k)$ is the wronskian
associated to the couple $\left\{  \chi_{+}^{h}\left(  \cdot,-k\right)
,\chi_{-}^{h}\left(  \cdot,k\right)  \right\}  $. In particular, the
inequality%
\begin{equation}
\frac{1}{\left\vert w^{h}(k)\right\vert }\leq\frac{h}{\left\vert k\right\vert
}\,, \label{wronskian_id}%
\end{equation}
implies that the maps $z\rightarrow\mathcal{G}^{z,h}(x,y)$, $z\rightarrow
\mathcal{H}^{z,h}(x,y)$ continuously extend up to the branch cut, both in the
limits $z\rightarrow k^{2}\pm i0$, with the only possible exception of the
point $z=0$.

The above characterization and the definition (\ref{q_z}) imply that
$z\rightarrow\left(  B_{\theta}\,q(z,h)-A_{\theta}\right)  $ is meromorphic
matrix-valued map in $\mathbb{C}\backslash\mathbb{R}_{+}$ with continuous
extension to $z\rightarrow k^{2}\pm i0$ for $k\neq0$. Due to the identity
(\ref{krein_1}), the conditions: $z\in res\left(  H_{0}^{h}\right)  $ and
$0\notin res\left(  B_{\theta}\,q(z,h)-A_{\theta}\right)  $ (i.e.: $z$ is a
pole for the inverse, matrix-valued, function $z\rightarrow\left(  B_{\theta
}\,q(z,h)-A_{\theta}\right)  ^{-1}$), compel: $z\in\sigma_{p}\left(
H_{\theta}^{h}\right)  $. Nevertheless, according to the result of the
Proposition \ref{Proposition_spectrum}, for defined positive potentials it
results: $res\left(  H_{\theta}^{h}\right)  =res\left(  H_{0}^{h}\right)
=\mathbb{C}\backslash\mathbb{R}_{+}$ provided that: $\theta\in B_{\delta
}\left(  0\right)  $, for a small $\delta>0$ possibly depending on $h$. Hence,
under these conditions, the inverse $\left(  B_{\theta}\,q(z,h)-A_{\theta
}\right)  ^{-1}$ exists in $\mathbb{C}\backslash\mathbb{R}_{+}$ and has
continuous extensions to the branch cut both in the limits $z=k^{2}\pm i0$,
with the only possible exception of the origin.

The generalized eigenfunctions of our model, next denoted with $\psi_{\theta
}^{h}(\cdot,k)$, solve of the boundary value problem%
\begin{equation}
\left\{
\begin{array}
[c]{l}%
\left(  -h^{2}\partial_{x}^{2}+\mathcal{V}\right)  u=k^{2}u\,,\qquad
x\in\mathbb{R}\backslash\left\{  a,b\right\}  \,,\ k\in\mathbb{R}\,,\\
\\
\smallskip e^{-\theta/2}u(b^{+})=u(b^{-})\,,\quad e^{-\theta\,3/2}u^{\prime
}(b^{+})=u^{\prime}(b^{-})\,,\\
\\
e^{-\theta/2}u(a^{-})=u(a^{+}),\quad e^{-\theta\,3/2}u^{\prime}(a^{-}%
)=u^{\prime}(a^{+})\,,
\end{array}
\right.  \label{gen_eigenfun_eq}%
\end{equation}
and fulfill the exterior conditions%
\begin{equation}
\psi_{\theta}^{h}(x,k)\left\vert _{\substack{x<a\\k>0}}\right.  =e^{i\frac
{k}{h}x}+R^{h}(k,\theta)e^{-i\frac{k}{h}x}\,,\quad\psi_{\theta}^{h}%
(x,k)\left\vert _{\substack{x>b\\k>0}}\right.  =T^{h}(k,\theta)e^{i\frac{k}%
{h}x}\,, \label{gen_eigenfun_ext1}%
\end{equation}%
\begin{equation}
\psi_{\theta}^{h}(x,k)\left\vert _{\substack{x<a\\k<0}}\right.  =T^{h}%
(k,\theta)e^{i\frac{k}{h}x}\,,\quad\psi_{\theta}^{h}(x,k)\left\vert
_{\substack{x>b\\k<0}}\right.  =e^{i\frac{k}{h}x}+R^{h}(k,\theta)e^{-i\frac
{k}{h}x}\,, \label{gen_eigenfun_ext2}%
\end{equation}
describing an incoming wave function of momentum $k$ with reflection and
transmission coefficients $R^{h}$ and $T^{h}$. For $\theta=0$, the generalized
eigenfunctions of $H_{0}^{h}$, $\psi_{0}^{h}(\cdot,k)$, depend on the Jost's
solutions $\chi_{\pm}^{h}$ according to%
\begin{equation}
\psi_{0}^{h}(x,k)=\left\{
\begin{array}
[c]{lll}%
-\frac{2ik}{hw^{h}(k)}\chi_{+}^{h}(x,k)\,, &  & \text{for }k\geq0\,,\\
&  & \\
\frac{2ik}{hw^{h}(-k)}\chi_{-}^{h}(x,-k)\,, &  & \text{for }k<0\,.
\end{array}
\right.  \label{gen_eigenfun_h_jost}%
\end{equation}
Following an approach similar to the one leading to the Krein-like resolvent
formula (\ref{krein_1}), an expansion for the difference: $\left.
\psi_{\theta}^{h}(x,k)-\psi_{0}^{h}(x,k)\right.  $ as $\left.  \theta
\rightarrow0\right.  $ has been provided with In \cite[eq. (2.19) and
(2.26)]{Man2}. Let $G^{k,h}$ and $H^{k,h}$ be defined by%
\begin{equation}%
\begin{array}
[c]{ccc}%
G^{\pm\left\vert k\right\vert ,h}\left(  \cdot,y\right)  =\lim_{z\rightarrow
k^{2}\pm i0}\mathcal{G}^{z,h}\left(  \cdot,y\right)  \,, &  & H^{\pm\left\vert
k\right\vert ,h}\left(  \cdot,y\right)  =\lim_{z\rightarrow k^{2}\pm
i0}\mathcal{H}^{z,h}\left(  \cdot,y\right)  \,,
\end{array}
\label{GH_k}%
\end{equation}
and denote with $g_{k,h,j}$ and $\mathcal{M}^{h}$ the corresponding limits of
$\gamma_{z,h,j}$ and $\left(  B_{\theta}\,q(z,h)-A_{\theta}\right)  $ (see the
definitions (\ref{defect1}), (\ref{q_z}) and (\ref{AB_theta_matrix})); namely,
we set%
\begin{equation}%
\begin{array}
[c]{ccc}%
g_{\pm\left\vert k\right\vert ,h,j}=\lim_{z\rightarrow k^{2}\pm i0}%
\gamma_{z,h,j}\,, &  & \mathcal{M}^{h}\left(  \pm\left\vert k\right\vert
,\theta\right)  =\lim_{z\rightarrow k^{2}\pm i0}\left(  B_{\theta
}\,q(z,h)-A_{\theta}\right)  \,.
\end{array}
\label{krein_coeff_k}%
\end{equation}
Due to (\ref{G_z_h})-(\ref{H_z_h}), $G^{k,h}$ and $H^{k,h}$ explicitly write
as%
\begin{align}
G^{k,h}\left(  \cdot,y\right)   &  =\frac{1}{h^{2}w^{h}(k)}\mathcal{\,}%
\left\{
\begin{array}
[c]{c}%
\chi_{+}^{h}\left(  \cdot,k\right)  \chi_{-}^{h}\left(  y,k\right)  \,,\ x\geq
y\,,\\
\\
\chi_{-}^{h}\left(  \cdot,k\right)  \chi_{+}^{h}\left(  y,k\right)
\,,\ x<y\,,
\end{array}
\right. \label{G_k_jost}\\
& \nonumber\\
H^{k,h}\left(  \cdot,y\right)   &  =\frac{-1}{h^{2}w^{h}(k)}\mathcal{\,}%
\left\{
\begin{array}
[c]{c}%
\chi_{+}^{h}\left(  \cdot,k\right)  \,\partial_{1}\chi_{-}^{h}\left(
y,k\right)  \,,\ x\geq y\,,\\
\\
\chi_{-}^{h}\left(  \cdot,k\right)  \,\partial_{1}\chi_{+}^{h}\left(
y,k\right)  \,,\ x<y\,,
\end{array}
\right.  \label{H_k_jost}%
\end{align}
while, according to the previous remarks, $g_{k,h,j}$ and $\mathcal{M}%
^{h}\left(  k,\theta\right)  $ are well defined and continuous w.r.t.
$k\in\mathbb{R}$, with the only possible exception of the origin. Denoting
with%
\begin{equation}
S^{h}\left(  \theta\right)  =\left\{  k\in\mathbb{R\,}\left\vert
\ \det\mathcal{M}^{h}\left(  k,\theta\right)  =0\right.  \right\}  \,.
\label{Singular_Set_h_theta}%
\end{equation}
the set of the singular points of $\mathcal{M}^{h}\left(  k,\theta\right)  $,
the representation (see \cite[Proposition 2.2]{Man2})%
\begin{align}
&  \psi_{\theta}^{h}(\cdot,k)\label{gen_eigenfun_Krein_h}\\
&  =\left\{
\begin{array}
[c]{lll}%
\psi_{0}^{h}(\cdot,k)-\sum_{i,j=1}^{4}\left[  \left(  \mathcal{M}^{h}\left(
k,\theta\right)  \right)  ^{-1}B_{\theta}\right]  _{ij}\Gamma_{j}%
^{k,h}\,g_{k,h,j}\,, &  & \text{for }k>0\,,\\
&  & \\
\psi_{0}^{h}(\cdot,k)-\sum_{i,j=1}^{4}\left[  \left(  \mathcal{M}^{h}\left(
-k,\theta\right)  \right)  ^{-1}B_{\theta}\right]  _{ij}\,\Gamma_{j}%
^{k,h}g_{-k,h,j}\,, &  & \text{for }k<0\,,
\end{array}
\right. \nonumber
\end{align}
holds for any fixed $h>0$, $\theta\in\mathbb{C}$ and $k\in\mathbb{R}^{\ast
}\backslash S^{h}\left(  \theta\right)  $, being $\Gamma^{k,h}$ the vector of
the boundary values%
\begin{equation}
2\Gamma^{k,h}=\left(  \psi_{0}^{h}(b,k),\partial_{1}\psi_{0}^{h}(b,k),\psi
_{0}^{h}(a,k),\partial_{1}\psi_{0}^{h}(a,k)\right)  \,.
\label{gen_eigenfun_Krein_h_1}%
\end{equation}

\subsection{\label{Sec_trace}Trace estimates}

We aim to control the coefficients at the r.h.s. of
(\ref{gen_eigenfun_Krein_h}) when both $\theta$ and $h$ are small. This
requires accurate estimates for the boundary values of $g_{k,h,j}$ (occurring
in the definition of the matrix $\mathcal{M}^{h}\left(  k,\theta\right)  $)
and $\psi_{0}^{h}(\cdot,k)$. In \cite[eq. (2.19) and (2.26)]{Man2} these
estimates have been provided for a finite energy range when the potential
describes quantum wells in a semiclassical island. We next reconsider this
problem under a generic condition of positivity for $1_{\left[  a,b\right]
}\mathcal{V}$. To this aim, we next recall some standard energy estimates; let
consider the problem
\begin{equation}
\left\{
\begin{array}
[c]{l}%
\left(  -h^{2}\partial_{x}^{2}+\mathcal{V}-\zeta^{2}\right)  u=0\,,\qquad
\text{in }\left(  a,b\right)  \,,\\
\\
\left[  h\partial_{x}+i\zeta\right]  u(a)=\gamma_{a}\,,\quad\left[
h\partial_{x}-i\zeta\right]  u(b)=\gamma_{b}\,,
\end{array}
\right.  \label{Ag_eq}%
\end{equation}
where: $\mathcal{V}\in L^{\infty}\left(  \left(  a,b\right)  ,\mathbb{R}%
\right)  $, $\gamma_{a},\gamma_{b}\in\mathbb{C}$, and $h>0$.

\begin{lemma}
\label{Lemma_energy_est}Assume $\zeta\in\overline{\mathbb{C}^{+}}$ such that:
$\mathcal{V}-\operatorname{Re}\zeta^{2}>c$ for some $c>0$. The solution of
(\ref{Ag_eq}) fulfills the estimate%
\begin{equation}
h^{\frac{1}{2}}\sup_{\left[  a,b\right]  }\left\vert u\right\vert +\left\Vert
hu^{\prime}\right\Vert _{L^{2}\left(  \left[  a,b\right]  \right)
}+\left\Vert u\right\Vert _{L^{2}\left(  \left[  a,b\right]  \right)  }\leq
C_{a,b,c}\frac{1}{h^{\frac{1}{2}}}\left(  \left\vert \gamma_{a}\right\vert
+\left\vert \gamma_{b}\right\vert \right)  \,, \label{energy_est1}%
\end{equation}
with $C_{a,b,c}>0$ possibly depending on the data.
\end{lemma}

\begin{proof}
From the equation%
\begin{equation}
\left\langle u,\left(  -h^{2}\partial_{x}^{2}+\mathcal{V}-\zeta^{2}\right)
u\right\rangle =0\,,
\end{equation}
an integration by parts yields%
\begin{equation}
\left\Vert hu^{\prime}\right\Vert _{L^{2}\left(  \left[  a,b\right]  \right)
}^{2}+\int_{a}^{b}\left(  \mathcal{V}-\zeta^{2}\right)  \left\vert
u\right\vert ^{2}\,dx+h^{2}\left(  u^{\ast}u^{\prime}(a)-u^{\ast}u^{\prime
}(b)\right)  =0\,.
\end{equation}
Taking into account the boundary conditions in (\ref{Ag_eq}), our assumptions
($\mathcal{V}-\operatorname{Re}\zeta^{2}>c$ and $\operatorname{Im}\zeta\geq0$)
imply%
\begin{equation}
\left\Vert hu^{\prime}\right\Vert _{L^{2}\left(  \left[  a,b\right]  \right)
}^{2}+c\left\Vert u\right\Vert _{L^{2}\left(  \left[  a,b\right]  \right)
}^{2}\leq h\operatorname{Re}\left(  \left\vert u^{\ast}(a)\right\vert
\left\vert \gamma_{a}\right\vert -\left\vert u^{\ast}(b)\right\vert \left\vert
\gamma_{b}\right\vert \right)  \,. \label{energy_est0}%
\end{equation}
The estimate (\ref{energy_est1}) then follows from (\ref{energy_est0}) by
taking into account the Gagliardo-Nirenberg inequality: $\sup_{\left[
a,b\right]  }\left\vert \varphi\right\vert \leq C_{b-a}\left\Vert
\varphi^{\prime}\right\Vert _{L^{2}\left(  \left(  a,b\right)  \right)
}^{\frac{1}{2}}\left\Vert \varphi\right\Vert _{L^{2}\left(  \left(
a,b\right)  \right)  }^{\frac{1}{2}}$.
\end{proof}

When the differential operator $\left(  -h^{2}\partial_{x}^{2}+\mathcal{V-}%
k^{2}\right)  $ is defined with a potential $\mathcal{V}\in L^{\infty}\left(
\mathbb{R},\mathbb{R}\right)  $ compactly supported on $\left[  a,b\right]  $,
the corresponding Green's functions solve boundary value problems of the type
(\ref{Ag_eq}) and the Lemma \ref{Lemma_energy_est} applies if: $\mathcal{V}%
-k^{2}>c$. Global-in-$k$ estimates for their boundary values are next
considered by combining the explicit representations in terms of the Jost's
solutions, given in (\ref{G_k_jost})-(\ref{H_k_jost}) and
(\ref{gen_eigenfun_h_jost}), and energy estimates in the low-energy regime. To
this aim, the potential is next assumed to fulfill the stronger condition%
\begin{equation}%
\begin{array}
[c]{ccccc}%
\mathcal{V}\in L^{\infty}(\mathbb{R},\mathbb{R})\,, &  & \text{supp\thinspace
}\mathcal{V}=\left[  a,b\right]  \,, &  & 1_{\left[  a,b\right]  }%
\mathcal{V}>c\,,
\end{array}
\label{V_pos1}%
\end{equation}
holding for same $c>0$.

\begin{proposition}
\label{Proposition_trace_est_h}Let $\mathcal{V}$ fulfills the conditions
(\ref{V_pos1}); the relations%
\begin{align}
&  \left.  \left\vert \left(  1+k\right)  \psi_{0}^{h}(y,k)\right\vert
+h\,\left\vert \partial_{1}\psi_{0}^{h}(y,k)\right\vert \leq\left\vert
k\right\vert \,C_{a,b,c}\,,\right. \label{Trace_est1}\\
& \nonumber\\
&  \left.  \left\vert \left(  1+k\right)  G^{k,h}\left(  y,y^{\prime}\right)
\right\vert +\left\vert 1_{\left[  a,b\right]  }H^{k,h}\left(  y,y^{\prime
}\right)  \right\vert +h\,\left\vert k\right\vert ^{-1}\left\vert \partial
_{1}H^{k,h}\left(  y,y^{\prime}\right)  \right\vert \leq C_{a,b,c}%
h^{-2}\,,\right.  \label{Trace_est2}%
\end{align}
hold for $y,y^{\prime}\in\left\{  a,b\right\}  $ and $k\in\mathbb{R}$, being
$C_{a,b,c}>0$ possibly depending on the data and $h\in\left(  0,h_{0}\right]
$ with $h_{0}>0$ small.
\end{proposition}

\begin{proof}
From the result of the Lemma \ref{Lemma_Jost}, the Jost's solutions $\chi
_{\pm}^{h}\left(  x,k\right)  $ are $\mathcal{C}_{x}^{1}$-continuous and the
exterior conditions (\ref{Jost_sol_ext_h}) can be used for the explicit
computation of $\partial_{1}^{j}\chi_{\pm}^{h}\left(  y,k\right)  $ when
$y=a,b$ and $j=0,1$. Then, the relations (\ref{gen_eigenfun_h_jost}),
(\ref{G_k_jost})-(\ref{H_k_jost}) allow to obtain (almost) explicit
representations of the quantities considered in (\ref{Trace_est1}%
)-(\ref{Trace_est2}). Let start considering the boundary values $\partial
_{1}^{j}\psi_{0}^{h}(y,k)$; we focus on the case $y=a$, while similar
computations hold for $y=b$. If $k<0$, the representation
(\ref{gen_eigenfun_h_jost}) and the exterior conditions (\ref{Jost_sol_ext_h})
imply%
\begin{equation}%
\begin{array}
[c]{ccc}%
1_{\left\{  k<0\right\}  }\left(  k\right)  \psi_{0}^{h}(a,k)=\frac
{2ik}{hw^{h}(-k)}e^{i\frac{k}{h}a}\,, & \text{and} & 1_{\left\{  k<0\right\}
}\left(  k\right)  \partial_{1}\psi_{0}^{h}(a,k)=-\frac{2k^{2}}{h^{2}%
w^{h}(-k)}e^{i\frac{k}{h}a}\,.
\end{array}
\label{gen_eigenfun_boundary}%
\end{equation}
Recalling that $w^{h}(-k)=\left(  w^{h}(k)\right)  ^{\ast}$, the inequality
(\ref{wronskian_id}) leads to%
\begin{equation}%
\begin{array}
[c]{ccc}%
\left\vert 1_{\left\{  k<0\right\}  }\left(  k\right)  \psi_{0}^{h}%
(a,k)\right\vert \leq2\,, & \text{and} & \left\vert 1_{\left\{  k<0\right\}
}\left(  k\right)  \partial_{1}\psi_{0}^{h}(a,k)\right\vert \leq2\left\vert
k\right\vert \,.
\end{array}
\label{gen_eigenfun_boundary_est1}%
\end{equation}
If $k\geq0$, comparing (\ref{gen_eigenfun_ext1}) with
(\ref{gen_eigenfun_h_jost}) and taking into account (\ref{Jost_sol_ext_h}), we
get the relation%
\begin{equation}
1_{\left\{  k\geq0\right\}  }\left(  k\right)  T^{h}(k,0)=-\frac{2ik}%
{hw^{h}(k)}\,. \label{T_h_k}%
\end{equation}
The identity: $\left\vert T^{h}(k,0)\right\vert ^{2}+\left\vert R^{h}%
(k,0)\right\vert ^{2}=1$, and the representations
\begin{equation}%
\begin{array}
[c]{ccc}%
1_{\left\{  k\geq0\right\}  }\left(  k\right)  \psi_{0}^{h}(a,k)=e^{i\frac
{k}{h}a}+R^{h}(k,0)e^{-i\frac{k}{h}a}\,, &  & 1_{\left\{  k\geq0\right\}
}\left(  k\right)  \partial_{1}\psi_{0}^{h}(a,k)=i\frac{k}{h}\left(
e^{i\frac{k}{h}a}-R^{h}(k,0)e^{-i\frac{k}{h}a}\right)  \,,
\end{array}
\label{gen_eigenfun_bound_est1_1}%
\end{equation}
yield%
\begin{equation}%
\begin{array}
[c]{ccc}%
\left\vert 1_{\left\{  k\geq0\right\}  }\left(  k\right)  \psi_{0}%
^{h}(a,k)\right\vert \leq2\,, &  & \left\vert 1_{\left\{  k\geq0\right\}
}\left(  k\right)  \partial_{1}\psi_{0}^{h}(a,k)\right\vert \leq2\frac{k}%
{h}\,.
\end{array}
\label{gen_eigenfun_boundary_est2}%
\end{equation}
From (\ref{gen_eigenfun_boundary_est1}), (\ref{gen_eigenfun_boundary_est2}),
and from similar computations in the case of $y=b$, follows%
\begin{equation}
\left\vert \psi_{0}^{h}(y,k)\right\vert +h\,\left\vert k\right\vert
^{-1}\left\vert \partial_{1}\psi_{0}^{h}(y,k)\right\vert \leq C_{a,b,c}%
\,,\quad y=a,b\,. \label{gen_eigenfun_boundary_est3}%
\end{equation}

The relations (\ref{Trace_est2}) are next considered for $y^{\prime}=a$ (when
$y^{\prime}=b$ the result follows from similar computations). According to
(\ref{gen_eigenfun_h_jost}), we have%
\begin{align}
&  \left.  \left\vert 1_{\left\{  k\geq0\right\}  }\left(  k\right)  \left(
hw^{h}\left(  k\right)  \right)  ^{-1}\partial_{1}^{j}\chi_{+}^{h}\left(
a,k\right)  \right\vert +\left\vert 1_{\left\{  k<0\right\}  }\left(
k\right)  \left(  hw^{h}\left(  k\right)  \right)  ^{-1}\partial_{1}^{j}%
\chi_{-}^{h}\left(  b,-k\right)  \right\vert \right. \\
& \nonumber\\
&  \left.  =\left\vert \left(  2k\right)  ^{-1}1_{\left\{  k\geq0\right\}
}\partial_{1}^{j}\psi_{0}^{h}(a,k)\right\vert +\left\vert \left(  2k\right)
^{-1}1_{\left\{  k<0\right\}  }\partial_{1}^{j}\psi_{0}^{h}(b,k)\right\vert
\right.  \,,\nonumber
\end{align}
then, the estimate (\ref{gen_eigenfun_boundary_est3}) lead us to%
\begin{equation}
\left\vert 1_{\left\{  k\geq0\right\}  }\left(  k\right)  \left(
hw^{h}\left(  k\right)  \right)  ^{-1}\partial_{1}^{j}\chi_{+}^{h}\left(
a,k\right)  \right\vert +\left\vert 1_{\left\{  k<0\right\}  }\left(
k\right)  \left(  hw^{h}\left(  k\right)  \right)  ^{-1}\partial_{1}^{j}%
\chi_{-}^{h}\left(  b,-k\right)  \right\vert \leq\frac{C_{a,b,c}}{2\left\vert
k\right\vert }\left(  \frac{\left\vert k\right\vert }{h}\right)  ^{j}\,,
\end{equation}
with $j=0,1$. Both these relations easily extend to all $k<0$ by recalling
that: $w^{h}(-k)=\left(  w^{h}(k)\right)  ^{\ast}$ and $\partial_{1}^{j}%
\chi_{\pm}^{h}\left(  \cdot,-k\right)  =\left(  \partial_{1}^{j}\chi_{\pm}%
^{h}\left(  \cdot,k\right)  \right)  ^{\ast}$. Finally, the identities:
$\partial_{1}^{j}\chi_{+}^{h}\left(  b,k\right)  =\left(  i\frac{k}{h}\right)
^{j}e^{i\frac{k}{h}b}$ and $\partial_{1}^{j}\chi_{-}^{h}\left(  a,k\right)
=\left(  -i\frac{k}{h}\right)  ^{j}e^{-i\frac{k}{h}a}$ (arising from
(\ref{Jost_sol_ext_h})) and the inequality (\ref{wronskian_id}) allow to
generalize this result to all $y=a,b$. This resumes as follows%
\begin{equation}
\left\vert \left(  hw^{h}\left(  k\right)  \right)  ^{-1}\partial_{1}^{j}%
\chi_{\pm}^{h}\left(  y,k\right)  \right\vert \leq\frac{C_{a,b,c}}{2\left\vert
k\right\vert }\left(  \frac{\left\vert k\right\vert }{h}\right)  ^{j}\,,\quad
y=a,b\,,\ j=0,1\,. \label{gen_eigenfun_boundary_est4}%
\end{equation}
From the representation (\ref{G_k_jost}), the boundary values $G^{\pm
\left\vert k\right\vert ,h}\left(  y,a\right)  $ write as
\begin{equation}
G^{k,h}\left(  y,a\right)  =\frac{1}{h^{2}w^{h}\left(  k\right)  }%
\mathcal{\,}\chi_{+}^{h}\left(  y,k\right)  e^{-i\frac{k}{h}a}\,,\quad
\text{for }y\in\left\{  a,b\right\}  \,. \label{G_k_h_boundary}%
\end{equation}
Let $c>0$ be such that (\ref{V_pos1}) holds at consider at first the case
$k^{2}\geq c$; using (\ref{gen_eigenfun_boundary_est4}) with $j=0$, we get%
\begin{equation}
\left\vert 1_{\left\{  k^{2}\geq c\right\}  }\left(  k\right)  G^{k,h}\left(
y,a\right)  \right\vert \leq1_{\left\{  k^{2}\geq c\right\}  }\left(
k\right)  \frac{C_{a,b,c}}{h\left\vert k\right\vert }\leq\frac{\tilde
{C}_{a,b,c}}{h\left(  1+\left\vert k\right\vert \right)  }\,,\quad y=a,b\,,
\label{G_k_h_boundary_est1}%
\end{equation}
for a suitable $\tilde{C}_{a,b,c}>0$ (depending on $a,b,c$). The function
$G^{\pm\left\vert k\right\vert ,h}\left(  \cdot,a\right)  $ solves an equation
of the type-(\ref{Ag_eq}) with: $\gamma_{a}=-\frac{1}{h}$ and $\gamma_{b}=0$;
when $k^{2}\leq c$, the lemma \ref{Lemma_energy_est} applies, allowing to
control the boundary values $G^{\pm\left\vert k\right\vert ,h}\left(
y,a\right)  $ according to%
\begin{equation}
\left\vert 1_{\left\{  k^{2}<c\right\}  }\left(  k\right)  G^{k,h}\left(
y,a\right)  \right\vert \leq C_{a,b,c}h^{-2}\,,\qquad y=a,b\,.
\label{G_k_h_boundary_est3}%
\end{equation}
Hence, (\ref{G_k_h_boundary_est1})-(\ref{G_k_h_boundary_est3}) lead us to%
\begin{equation}
\left\vert \left(  1+\left\vert k\right\vert \right)  G^{k,h}\left(
y,a\right)  \right\vert \leq K_{a,b,c}h^{-2}\,,\qquad y=a,b\,.
\label{G_k_h_boundary_est}%
\end{equation}
Using the representation (\ref{H_k_jost}), the exterior condition:
$1_{\left\{  x<a\right\}  }\left(  x\right)  \chi_{-}^{h}\left(  x,k\right)
=e^{-i\frac{k}{h}x}$ and the $\mathcal{C}_{x}^{1}$-regularity of $\chi_{-}%
^{h}\left(  \cdot,k\right)  $, it follows%
\begin{equation}
1_{\left[  a,b\right]  }H^{k,h}\left(  y,a\right)  =\frac{-ik}{h^{3}%
w^{h}\left(  k\right)  }\mathcal{\,}\chi_{+}^{h}\left(  y,k\right)
\,e^{-i\frac{k}{h}a}\,,\quad y=a,b\,, \label{H_k_h_boundary}%
\end{equation}
and the inequality (\ref{gen_eigenfun_boundary_est4}), $j=0$, yields%
\begin{equation}
\left\vert H^{k,h}\left(  y,a\right)  \right\vert \leq\frac{C_{a,b,c}}{h^{2}%
}\,,\qquad y=a,b\,. \label{H_k_h_boundary_est}%
\end{equation}
Using once more (\ref{H_k_jost}), we have%
\begin{equation}
\partial_{1}H^{k,h}\left(  y,a\right)  =\frac{-ik}{h^{3}w^{h}\left(  k\right)
}\mathcal{\,}\partial_{1}\chi_{+}^{h}\left(  y,k\right)  \,e^{-i\frac{k}{h}%
a}\,,\quad y=a,b\,. \label{D_H_k_h_boundary}%
\end{equation}
Then, (\ref{gen_eigenfun_boundary_est4}), $j=1$, implies%
\begin{equation}
\left\vert \partial_{1}H^{k,h}\left(  a,a\right)  \right\vert \leq
C_{a,b,c}\frac{\left\vert k\right\vert }{h^{3}}\,,\quad y=a,b\,.
\label{D_H_k_h_boundary_est}%
\end{equation}
For $y^{\prime}=a$, the inequality (\ref{Trace_est2}) follows from
(\ref{G_k_h_boundary_est}), (\ref{H_k_h_boundary_est}) and
(\ref{D_H_k_h_boundary_est}) with a suitable $C_{a,b,c}$.

Finally, we reconsider the bound (\ref{gen_eigenfun_boundary_est3}); according
to (\ref{G_k_jost}), the relation (\ref{Trace_est2}) yields%
\begin{equation}
G^{k,h}\left(  b,a\right)  =\left(  h^{2}w^{h}\left(  k\right)  \right)
^{-1}\mathcal{\,}e^{i\frac{k}{h}\left(  b-a\right)  }=\mathcal{O}\left(
h^{-2}\right)  \,, \label{Jost_fun_boundary_est}%
\end{equation}
uniformly w.r.t. $k\in\mathbb{R}$. It follows: $\inf_{\mathbb{R}}\left\vert
w^{h}\left(  k\right)  \right\vert >c_{0}$ for a suitable $c_{0}>0$ possibly
depending on the data, while, taking into account (\ref{wronskian_id}), we
get: $\left(  w^{h}\left(  k\right)  \right)  ^{-1}=\mathcal{O}\left(  \left(
1+k\right)  ^{-1}\right)  $. Hence, the relations $(i)$%
-(\ref{gen_eigenfun_boundary}), (\ref{T_h_k}) and $(i)$%
-(\ref{gen_eigenfun_bound_est1_1}), yield: $\left\vert \psi_{0}^{h}%
(a,k)\right\vert =\mathcal{O}\left(  k\left(  1+k\right)  ^{-1}\right)  $;
this improves the previous estimates according to (\ref{Trace_est1}).
\end{proof}

\begin{remark}
\label{Remark_trace}The result presented in \ref{Proposition_trace_est_h}
stands upon the regularity of the Jost's solution at the boundaries points
$\left\{  a,b\right\}  $. This property, considered in the Lemma
\ref{Lemma_Jost}, generically holds for positive defined and compactly
supported potentials, while the trace estimates (\ref{Trace_est1}%
)-(\ref{Trace_est2}) do not depend on the particular shape of $\mathcal{V}$,
provided that it fulfills the conditions (\ref{V_pos1}).
\end{remark}

\subsection{Generalized eigenfunctions expansion}

The result of the Proposition \ref{Proposition_trace_est_h} can be implemented
to obtain an expansion of the\ modified generalized eigenfunctions
$\psi_{\theta}^{h}(\cdot,k)$ when both $\theta$ and $h$ are small. Using the
notation introduced in (\ref{GH_k}) and (\ref{krein_coeff_k}), a direct
computation leads to%
\begin{align}
&  \left.  \medskip\mathcal{M}^{h}\left(  k,\theta\right)  =\right. \\
&  \left.  \medskip%
\begin{pmatrix}
\mathcal{\beta}\left(  \theta\right)  H^{k,h}(b^{+},b) & -\mathcal{\beta
}\left(  \theta\right)  \partial_{1}H^{k,h}(b,b) & \mathcal{\beta}\left(
\theta\right)  H^{k,h}(a,b) & -\mathcal{\beta}\left(  \theta\right)
\partial_{1}H^{k,h}(b,a)\\
\mathcal{\beta}\left(  \theta\right)  G^{k,h}(b,b) & -\mathcal{\beta}\left(
\theta\right)  H^{k,k}(b^{+},b) & \mathcal{\beta}\left(  \theta\right)
G^{k,h}(b,a) & -\mathcal{\beta}\left(  \theta\right)  H^{k,h}(b,a)\\
\mathcal{\beta}\left(  -\theta\right)  H^{k,h}(b,a) & -\mathcal{\beta}\left(
-\theta\right)  \partial_{1}H^{k,h}(a,b) & \mathcal{\beta}\left(
-\theta\right)  H^{k,h}(a^{-},a) & -\mathcal{\beta}\left(  -\theta\right)
\partial_{1}H^{k,h}(a,a)\\
\mathcal{\beta}\left(  -\theta\right)  G^{k,h}(a,b)\medskip & -\mathcal{\beta
}\left(  -\theta\right)  H^{k,h}(a,b) & \mathcal{\beta}\left(  -\theta\right)
G^{k,h}(a,a) & -\mathcal{\beta}\left(  -\theta\right)  H^{k,h}(a^{-},a)
\end{pmatrix}
\right. \nonumber\\
&  \left.  -\frac{1}{h^{2}}%
\begin{pmatrix}
\alpha\left(  \theta\right)  +\frac{h^{2}}{2}\mathcal{\beta}\left(
\theta\right)  &  &  & \\
& \alpha\left(  \theta\right)  -\frac{h^{2}}{2}\mathcal{\beta}\left(
\theta\right)  &  & \\
&  & \alpha\left(  -\theta\right)  -\frac{h^{2}}{2}\mathcal{\beta}\left(
-\theta\right)  & \\
&  &  & \alpha\left(  -\theta\right)  +\frac{h^{2}}{2}\mathcal{\beta}\left(
-\theta\right)
\end{pmatrix}
\,,\right.  \qquad\nonumber
\end{align}
where $\alpha\left(  \theta\right)  $ and $\mathcal{\beta}\left(
\theta\right)  $ are defined by%
\begin{equation}
\mathcal{\alpha}\left(  \theta\right)  =1+e^{\frac{\theta}{2}}\,,\qquad
\mathcal{\beta}\left(  \theta\right)  =1-e^{\frac{\theta}{2}}\,.
\end{equation}
As consequence of the estimates (\ref{Trace_est1})-(\ref{Trace_est2}), for
defined positive potentials the above relation rephrases as%
\begin{align}
&  \left.  \medskip\mathcal{M}^{h}\left(  k,\theta\right)  =\right.
\label{M_k_teta}\\
&  \medskip\left.  -\frac{1}{h^{2}}%
\begin{pmatrix}
\alpha\left(  \theta\right)  &  &  & \\
& \alpha\left(  \theta\right)  &  & \\
&  & \alpha\left(  -\theta\right)  & \\
&  &  & \alpha\left(  -\theta\right)
\end{pmatrix}
\right. \nonumber\\
&  \left.  +%
\begin{pmatrix}
\mathcal{\beta}\left(  \theta\right)  \mathcal{O}\left(  h^{-2}\right)  &
\mathcal{\beta}\left(  \theta\right)  \mathcal{O}\left(  \left\vert
k\right\vert h^{-3}\right)  & \mathcal{\beta}\left(  \theta\right)
\mathcal{O}\left(  h^{-2}\right)  & \mathcal{\beta}\left(  \theta\right)
\mathcal{O}\left(  \left\vert k\right\vert h^{-3}\right) \\
\mathcal{\beta}\left(  \theta\right)  \mathcal{O}\left(  h^{-2}\left(
1+\left\vert k\right\vert \right)  ^{-1}\right)  & \mathcal{\beta}\left(
\theta\right)  \mathcal{O}\left(  h^{-2}\right)  & \mathcal{\beta}\left(
\theta\right)  \mathcal{O}\left(  h^{-2}\left(  1+\left\vert k\right\vert
\right)  ^{-1}\right)  & \mathcal{\beta}\left(  \theta\right)  \mathcal{O}%
\left(  h^{-2}\right) \\
\mathcal{\beta}\left(  -\theta\right)  \mathcal{O}\left(  h^{-2}\right)  &
\mathcal{\beta}\left(  -\theta\right)  \mathcal{O}\left(  \left\vert
k\right\vert h^{-3}\right)  & \mathcal{\beta}\left(  -\theta\right)
\mathcal{O}\left(  h^{-2}\right)  & \mathcal{\beta}\left(  -\theta\right)
\mathcal{O}\left(  \left\vert k\right\vert h^{-3}\right) \\
\mathcal{\beta}\left(  -\theta\right)  \mathcal{O}\left(  h^{-2}\left(
1+\left\vert k\right\vert \right)  ^{-1}\right)  & \mathcal{\beta}\left(
-\theta\right)  \mathcal{O}\left(  h^{-2}\right)  & \mathcal{\beta}\left(
-\theta\right)  \mathcal{O}\left(  h^{-2}\left(  1+\left\vert k\right\vert
\right)  ^{-1}\right)  & \mathcal{\beta}\left(  -\theta\right)  \mathcal{O}%
\left(  h^{-2}\right)
\end{pmatrix}
\right.  \,.\nonumber
\end{align}
being the symbols $\mathcal{O}\left(  \cdot\right)  $ referred to the metric
space $\mathbb{R}\times\left(  0,h_{0}\right]  $ and defining continuous
functions of $k\in\mathbb{R}$. From the definition of $\mathcal{\alpha}\left(
\theta\right)  $, $\mathcal{\beta}\left(  \theta\right)  $, the coefficients
of $\mathcal{M}^{h}\left(  k,\theta\right)  $ result $\theta$-holomorphic and
continuous w.r.t. $\left(  k,\theta\right)  \in\mathbb{C}\times\mathbb{R}$
while, using the expansions: $\mathcal{\alpha}\left(  \theta\right)
=2+\mathcal{O}\left(  \theta\right)  $ and $\mathcal{\beta}\left(
\theta\right)  =\mathcal{O}\left(  \theta\right)  $, follows%
\begin{equation}
\mathcal{M}^{h}\left(  k,\theta\right)  =-\frac{2}{h^{2}}1_{\mathbb{C}^{4}%
}+\theta\,m^{h}\left(  k,\theta\right)  \,, \label{M_k_teta_exp}%
\end{equation}
where the remainder term is%
\begin{equation}
m^{h}\left(  k,\theta\right)  =%
\begin{pmatrix}
\mathcal{O}\left(  h^{-2}\right)  & \mathcal{O}\left(  \left\vert k\right\vert
h^{-3}\right)  & \mathcal{O}\left(  h^{-2}\right)  & \mathcal{O}\left(
\left\vert k\right\vert h^{-3}\right) \\
\mathcal{O}\left(  h^{-2}\left(  1+\left\vert k\right\vert \right)
^{-1}\right)  & \mathcal{O}\left(  h^{-2}\right)  & \mathcal{O}\left(
h^{-2}\left(  1+\left\vert k\right\vert \right)  ^{-1}\right)  &
\mathcal{O}\left(  h^{-2}\right) \\
\mathcal{O}\left(  h^{-2}\right)  & \mathcal{O}\left(  \left\vert k\right\vert
h^{-3}\right)  & \mathcal{O}\left(  h^{-2}\right)  & \mathcal{O}\left(
\left\vert k\right\vert h^{-3}\right) \\
\mathcal{O}\left(  h^{-2}\left(  1+\left\vert k\right\vert \right)
^{-1}\right)  & \mathcal{O}\left(  h^{-2}\right)  & \mathcal{O}\left(
h^{-2}\left(  1+\left\vert k\right\vert \right)  ^{-1}\right)  &
\mathcal{O}\left(  h^{-2}\right)
\end{pmatrix}
\,, \label{M_k_teta_rem}%
\end{equation}
Hence, $\mathcal{M}^{h}\left(  k,\theta\right)  $ is invertible for all $k$
provided that $\theta$ is small depending on $h$, while, from the
representation (\ref{gen_eigenfun_Krein_h}), an expansion of $\psi_{\theta
}^{h}(\cdot,k)$ for small values of $\theta$ follows.

\begin{proposition}
Assume $h\in\left(  0,h_{0}\right]  $, $\left\vert \theta\right\vert \leq
h^{2}$ and let $\mathcal{V}$ be defined according to (\ref{V_pos1}) for some
$c>0$. For a suitably small $h_{0}$, the solutions $\psi_{\theta}^{h}%
(\cdot,k)$ of the generalized eigenfunctions problem (\ref{gen_eigenfun_eq}),
(\ref{gen_eigenfun_ext1})-(\ref{gen_eigenfun_ext2}) allow the expansion%
\begin{align}
&  \left.  \psi_{\theta}^{h}(\cdot,k)-\psi_{0}^{h}(\cdot,k)=\right.
\label{gen_eigenfun_exp}\\
& \nonumber\\
&  \left.  \mathcal{O}\left(  \frac{\theta k}{h}\right)  G^{\left\vert
k\right\vert ,h}(\cdot,b)+\mathcal{O}\left(  \frac{\theta k}{1+k}\right)
H^{\left\vert k\right\vert ,h}(\cdot,b)+\mathcal{O}\left(  \frac{\theta k}%
{h}\right)  G^{\left\vert k\right\vert ,h}(\cdot,a)+\mathcal{O}\left(
\frac{\theta k}{1+k}\right)  H^{\left\vert k\right\vert ,h}(\cdot,a)\right.
\,,\nonumber
\end{align}
where the symbols $\mathcal{O}\left(  \cdot\right)  $ denote functions of the
variables $\left\{  \theta,k,h\right\}  \in$ $B_{h^{2}}(0)\times
\mathbb{R}\times\left(  0,h_{0}\right]  $ holomorphic w.r.t. $\theta$ and
continuous in $k$.
\end{proposition}

\begin{proof}
The coefficients of the remainder $m^{h}\left(  k,\theta\right)  $ in
(\ref{M_k_teta_exp})-(\ref{M_k_teta_rem}), depending on the variables
$\left\{  \theta,k,h\right\}  $, are $\mathcal{O}\left(  h^{-3}\right)  $;
hence, for $\left\vert \theta\right\vert \leq h^{2}$ and $h\in\left(
0,h_{0}\right]  $ with $h_{0}$ suitably small, the expansion
(\ref{M_k_teta_exp}), rephrasing as: $\mathcal{M}^{h}\left(  k,\theta\right)
=-2h^{-2}1_{\mathbb{C}^{4}}+\mathcal{O}\left(  h^{-1}\right)  $, defines an
invertible matrix for all $k\in\mathbb{R}$ and the representation
(\ref{gen_eigenfun_Krein_h}) globally holds. In particular, from
(\ref{M_k_teta_exp})-(\ref{M_k_teta_rem}), a direct computation yields%
\begin{equation}
\det\mathcal{M}^{h}\left(  k,\theta\right)  =h^{-8}\left(  16+\mathcal{O}%
\left(  h\right)  \right)  \,,
\end{equation}
and%
\begin{equation}
\left(  \mathcal{M}^{h}\left(  k,\theta\right)  \right)  ^{-1}=h^{2}%
\begin{pmatrix}
-1/2+\mathcal{O}\left(  h\right)  & \mathcal{O}\left(  hk\right)  &
\mathcal{O}\left(  h\right)  & \mathcal{O}\left(  hk\right) \\
\mathcal{O}\left(  h^{2}\left(  1+k\right)  ^{-1}\right)  & -1/2+\mathcal{O}%
\left(  h\right)  & \mathcal{O}\left(  h^{2}\left(  1+k\right)  ^{-1}\right)
& \mathcal{O}\left(  h\right) \\
\mathcal{O}\left(  h\right)  & \mathcal{O}\left(  hk\right)  &
-1/2+\mathcal{O}\left(  h\right)  & \mathcal{O}\left(  hk\right) \\
\mathcal{O}\left(  h^{2}\left(  1+k\right)  ^{-1}\right)  & \mathcal{O}\left(
h\right)  & \mathcal{O}\left(  h^{2}\left(  1+k\right)  ^{-1}\right)  &
-1/2+\mathcal{O}\left(  h\right)
\end{pmatrix}
\label{gen_eigenfun_Krein_coeff}%
\end{equation}
From the relations (\ref{Trace_est1}) and the definitions
(\ref{AB_theta_matrix}), (\ref{gen_eigenfun_Krein_h_1}), follows%
\begin{equation}
B_{\theta}\Gamma^{k,h}=\left\{  \,\mathcal{O}\left(  \frac{\theta k}%
{h}\right)  ,\ \mathcal{O}\left(  \frac{\theta k}{1+k}\right)
\,,\ \mathcal{O}\left(  \frac{\theta k}{h}\right)  \,,\ \mathcal{O}\left(
\frac{\theta k}{1+k}\right)  \right\}  \,. \label{gen_eigenfun_Krein_coeff1}%
\end{equation}
where the symbols $\mathcal{O}\left(  \cdot\right)  $ are referred to the
metric space $B_{h^{2}}(0)\times\mathbb{R}\times\left(  0,h_{0}\right]  $.
Making use of the above relations, we get%
\begin{equation}
\left(  \mathcal{M}^{h}\left(  k,\theta\right)  \right)  ^{-1}B_{\theta}%
\Gamma^{k,h}=\left\{  \,\mathcal{O}\left(  \frac{\theta k}{h}\right)
,\ \mathcal{O}\left(  \frac{\theta k}{1+k}\right)  \,,\ \mathcal{O}\left(
\frac{\theta k}{h}\right)  \,,\ \mathcal{O}\left(  \frac{\theta k}%
{1+k}\right)  \right\}  \label{gen_eigenfun_Krein_coeff2}%
\end{equation}
Then, the expansion (\ref{gen_eigenfun_exp}) follows from the formula
(\ref{gen_eigenfun_Krein_h})-(\ref{gen_eigenfun_Krein_h_1}) by taking into
account (\ref{gen_eigenfun_Krein_coeff2}) and the definition of $g_{k,h,j}$.
\end{proof}

\subsection{Stationary wave operators and uniform-in-time estimates for the
dynamical system}

Following \cite{Man1}, we next construct a similarity between $H_{\theta}^{h}$
and $H_{0}^{h}$ by making use of the stationary waves operators related to the
scattering system $\left\{  H_{\theta}^{h},H_{0}^{h}\right\}  $. Let us recall
that, for potentials defined as in (\ref{V}), the generalized Fourier
transform associated to $H_{0}^{h}$,%
\begin{equation}
\left(  \mathcal{F}_{\mathcal{V}}^{h}\varphi\right)  (k)=\int_{\mathbb{R}%
}\frac{dx}{\left(  2\pi h\right)  ^{1/2}}\,\left(  \psi_{0}^{h}\left(
x,k\right)  \right)  ^{\ast}\varphi(x)\,,\qquad\varphi\in L^{2}(\mathbb{R})\,,
\label{gen_Fourier_h}%
\end{equation}
is a bounded operator on $L^{2}(\mathbb{R})$ with a right inverse coinciding
with the adjoint $\left(  \mathcal{F}_{\mathcal{V}}^{h}\right)  ^{\ast}$%
\begin{equation}
\left(  \mathcal{F}_{\mathcal{V}}^{h}\right)  ^{\ast}f(x)=\int\frac
{dk}{\left(  2\pi h\right)  ^{1/2}}\,\psi_{0}^{h}(x,k)f(k)\,,
\label{gen_Fourier_h_inv}%
\end{equation}
and it results: $\mathcal{F}_{\mathcal{V}}^{h}\left(  \mathcal{F}%
_{\mathcal{V}}^{h}\right)  ^{\ast}=1_{L^{2}\left(  \mathbb{R}\right)  }$ in
$L^{2}(\mathbb{R})$, while the product $\left(  \mathcal{F}_{\mathcal{V}}%
^{h}\right)  ^{\ast}\mathcal{F}_{\mathcal{V}}^{h}$ defines the projector on
the absolutely continuous subspace of $H_{0}^{h}$ (cf. \cite{Yafa}). In
addition, when $\mathcal{V}$ is positive defined, $H_{0}^{h}$ has a purely
absolutely continuous spectrum coinciding with $\mathbb{R}_{+}$; in this case
$\mathcal{F}_{\mathcal{V}}$ is an unitary map with range $L^{2}(\mathbb{R})$
and the representation: $1_{L^{2}\left(  \mathbb{R}\right)  }=\left(
\mathcal{F}_{\mathcal{V}}^{h}\right)  ^{\ast}\left(  \mathcal{F}_{\mathcal{V}%
}^{h}\varphi\right)  $ holds. According to the above notation, the standard
Fourier transform operator, corresponding to the case $\mathcal{V}=0$, is next
denoted with $\mathcal{F}_{0}^{h}$. We consider the maps $\phi_{\alpha}^{h}$
and $\psi_{\alpha}^{h}$, acting on $L^{2}\left(  \mathbb{R}\right)  $ as%
\begin{align}
\phi_{\alpha}^{h}(\varphi,f)  &  =\int_{\mathbb{R}}\frac{dk}{\left(  2\pi
h\right)  ^{1/2}}\,\,f(k)\,G^{\left\vert k\right\vert ,h}\left(  \cdot
,\alpha\right)  \left(  \mathcal{F}_{\mathcal{V}}^{h}\varphi\right)
(k)\,,\quad\alpha\in\left\{  a,b\right\}  \,,\label{Phi_alpha_h}\\
& \nonumber\\
\psi_{\alpha}^{h}(\varphi,f)  &  =\int_{\mathbb{R}}\frac{dk}{\left(  2\pi
h\right)  ^{1/2}}\,\,g(k)\,H^{\left\vert k\right\vert ,h}(\cdot,\alpha)\left(
\mathcal{F}_{\mathcal{V}}^{h}\varphi\right)  (k)\,,\quad\alpha\in\left\{
a,b\right\}  \,. \label{Psi_alpha_h}%
\end{align}
Here $G^{k,h}$ and $H^{k,h}$ are the limits of the Green's functions on the
branch cut (see the definition in (\ref{G_k_jost})-(\ref{H_k_jost})), while
$f$ is an auxiliary function, possibly depending on $h$ and $\theta$ aside
from $k$.

\begin{lemma}
\label{Lemma_Simil_est}Let $h\in\left(  0,h_{0}\right]  $ and $\mathcal{V}$ be
defined according to (\ref{V_pos1}) with $h_{0}$ suitably small. Assume
$f_{j=1,2}\in L_{k}^{\infty}\left(  \mathbb{R}\right)  $ such that:
$f=\mathcal{O}\left(  k\right)  $ and $g=\mathcal{O}\left(  \frac{k}%
{1+k}\right)  $. Then it results%
\begin{equation}
h\left\Vert \phi_{\alpha}^{h}(\cdot,f)\right\Vert _{\mathcal{L}\left(
L^{2}\left(  \mathbb{R}\right)  ,\right)  }+\left\Vert \psi_{\alpha}^{h}%
(\cdot,g)\right\Vert _{\mathcal{L}\left(  L^{2}\left(  \mathbb{R}\right)
\right)  }\leq C_{a,b,c}h^{-2}\,, \label{phi_psi_est}%
\end{equation}
and%
\begin{equation}
h\left\Vert \phi_{\alpha}^{h}(\cdot,f)\right\Vert _{\mathcal{L}\left(
H^{2}\left(  \mathbb{R}\right)  ,H^{2}\left(  \mathbb{R}\backslash\left\{
a,b\right\}  \right)  \right)  }+\left\Vert \psi_{\alpha}^{h}(\cdot
,g)\right\Vert _{\mathcal{L}\left(  H^{2}\left(  \mathbb{R}\right)
,H^{2}\left(  \mathbb{R}\backslash\left\{  a,b\right\}  \right)  \right)
}\leq C_{a,b,c}h^{-2}\,, \label{phi_psi_est_1}%
\end{equation}
where $C_{a,b,c}$ is a positive constant depending on the data.
\end{lemma}

\begin{proof}
We show that each of the maps $\phi_{\alpha}^{h}\left(  \cdot,f\right)  $ and
$\psi_{\alpha}^{h}\left(  \cdot,g\right)  $, $\alpha=a,b$, can be expressed as
a superpositions of terms having the following form%
\begin{equation}
1_{\left\{  x\geq\alpha\right\}  }\mathcal{T}_{\mathcal{\alpha}}^{h}\left(
\mu_{1}+\mathcal{P\circ}\mu_{2}\right)  \mathcal{F}_{\mathcal{V}}%
^{h}+1_{\left\{  x<\alpha\right\}  }\mathcal{T}_{\mathcal{\alpha}}^{h}\left(
\mu_{3}+\mathcal{P\circ}\mu_{4}\right)  \mathcal{F}_{\mathcal{V}}^{h}\,,
\label{phi_psi_formula}%
\end{equation}
where $\mu_{i}\in L_{k}^{\infty}\left(  \mathbb{R}\right)  $, $\mathcal{T}%
_{\mathcal{\alpha}}^{h}=\left(  \mathcal{F}_{0}^{h}\right)  ^{\ast}$ or
$\mathcal{T}_{\mathcal{\alpha}}^{h}=\left(  \mathcal{F}_{\mathcal{V}}%
^{h}\right)  ^{\ast}$ depending on $\alpha=a,b$, while $\mathcal{P}$ denotes
the parity operator: $\mathcal{P}u(t)=u(-t)$. The estimate (\ref{phi_psi_est})
is a direct consequence of this representation. Let us focus on the case
$\alpha=b$ and explicitly consider $\phi_{b}^{h}(\cdot,f)$. As it follows from
(\ref{G_k_jost})-(\ref{gen_eigenfun_h_jost}), the functions $G^{\left\vert
k\right\vert ,h}\left(  \cdot,b\right)  $ and $H^{\left\vert k\right\vert
,h}\left(  \cdot,b\right)  $ allow the representations%
\begin{align}
&  \left.  G^{\left\vert k\right\vert ,h}\left(  x,b\right)  =1_{\left\{
x\geq b\right\}  }\left(  x\right)  \frac{1}{h^{2}w^{h}(\left\vert
k\right\vert )}e^{i\frac{\left\vert k\right\vert }{h}x}\chi_{-}^{h}\left(
b,\left\vert k\right\vert \right)  +1_{\left\{  x<b\right\}  }\left(
x\right)  \left(  \frac{-1}{2i\left\vert k\right\vert h}\right)  \psi_{0}%
^{h}(x,-\left\vert k\right\vert )e^{i\frac{\left\vert k\right\vert }{h}%
b}\right.  \,,\label{G_gen_eigenfun_k1}\\
& \nonumber\\
&  \left.  H^{\left\vert k\right\vert ,h}\left(  x,b\right)  =1_{\left\{
x\geq b\right\}  }\left(  x\right)  \frac{1}{h^{2}w^{h}(\left\vert
k\right\vert )}e^{i\frac{\left\vert k\right\vert }{h}x}\,\partial_{1}\chi
_{-}^{h}\left(  b,\left\vert k\right\vert \right)  \,+1_{\left\{  x<b\right\}
}\left(  x\right)  \frac{\left\vert k\right\vert }{2h^{2}k}\psi_{0}%
^{h}(x,-\left\vert k\right\vert )e^{i\frac{\left\vert k\right\vert }{h}%
b}\right.  \label{H_gen_eigenfun_k1}%
\end{align}
The condition $f\left(  k\right)  =\mathcal{O}\left(  k\right)  $ and the
estimates (\ref{gen_eigenfun_boundary_est4}) implies: $\left(  h^{2}%
w^{h}(k)\right)  ^{-1}\,f(k)\chi_{-}^{h}\left(  b,\left\vert k\right\vert
\right)  =$ $\mathcal{O}\left(  h^{-1}\right)  $ and $\frac{-1}{2i\left\vert
k\right\vert h}f(k)e^{i\frac{\left\vert k\right\vert }{h}b}=\mathcal{O}\left(
h^{-1}\right)  $; thus, using (\ref{G_gen_eigenfun_k1}) for $x\geq b$ we get%
\begin{align}
&  \left.  1_{\left\{  x\geq b\right\}  }\phi_{b}^{h}(\varphi,f)=\right.
\nonumber\\
& \nonumber\\
&  \left.  1_{\left\{  x\geq b\right\}  }\left(  \int_{0}^{+\infty}\frac
{dk}{\left(  2\pi h\right)  ^{1/2}}\,\mathcal{O}\left(  h^{-1}\right)
e^{i\frac{k}{h}x}\left(  \mathcal{F}_{\mathcal{V}}^{h}\varphi\right)
(k)+\int_{-\infty}^{0}\frac{dk}{\left(  2\pi h\right)  ^{1/2}}\,\mathcal{O}%
\left(  h^{-1}\right)  e^{-i\frac{k}{h}x}\left(  \mathcal{F}_{\mathcal{V}}%
^{h}\varphi\right)  (k)\right)  \,.\right.
\end{align}
The previous identity rephrases as%
\begin{equation}
1_{\left\{  x\geq b\right\}  }\phi_{b}^{h}(\varphi,f)=1_{\left\{  x\geq
b\right\}  }\left(  \mathcal{F}_{0}^{h}\right)  ^{\ast}\,\left[  1_{\left\{
k\geq0\right\}  }\left(  k\right)  \left(  \mathcal{O}\left(  h^{-1}\right)
+\mathcal{P\circ O}\left(  h^{-1}\right)  \right)  \mathcal{F}_{\mathcal{V}%
}^{h}\varphi\right]  \,, \label{simil_est_1}%
\end{equation}
where the symbols $\mathcal{O}\left(  \cdot\right)  $, denoting functions of
the variables $k$ and $h$, are defined in the sense of the metric space
$\mathbb{R}\times\left(  0,h_{0}\right]  $. Using (\ref{G_gen_eigenfun_k1})
for $x<b$ leads to%
\begin{align}
&  \left.  1_{\left\{  x<b\right\}  }\phi_{b}^{h}(\varphi,f)=\right.
\nonumber\\
& \nonumber\\
&  \left.  1_{\left\{  x<b\right\}  }\left(  \int_{0}^{+\infty}\frac
{dk}{\left(  2\pi h\right)  ^{1/2}}\,\mathcal{O}\left(  h^{-1}\right)
\psi_{0}^{h}(\cdot,-k)\left(  \mathcal{F}_{\mathcal{V}}^{h}\varphi\right)
(k)+\int_{-\infty}^{0}\frac{dk}{\left(  2\pi h\right)  ^{1/2}}\,\mathcal{O}%
\left(  h^{-1}\right)  \psi_{0}^{h}(\cdot,k)\left(  \mathcal{F}_{\mathcal{V}%
}^{h}\varphi\right)  (k)\right)  \right.  \,,
\end{align}
and, proceeding as before, we get%
\begin{equation}
1_{\left\{  x<b\right\}  }\phi_{b}^{h}(\varphi,f)=1_{\left\{  x<b\right\}
}\left(  \mathcal{F}_{\mathcal{V}}^{h}\right)  ^{\ast}\,\left[  1_{\left\{
k<0\right\}  }\left(  k\right)  \left(  \mathcal{P\circ O}\left(
h^{-1}\right)  +\mathcal{O}\left(  h^{-1}\right)  \right)  \mathcal{F}%
_{\mathcal{V}}^{h}\varphi\right]  \,. \label{simil_est_2}%
\end{equation}
From (\ref{simil_est_1}) and (\ref{simil_est_2}) we get a representation of
the type given in (\ref{phi_psi_formula})%
\begin{align}
&  \left.  \phi_{b}^{h}(\varphi,f)=1_{\left\{  x\geq b\right\}  }\left(
\mathcal{F}_{0}^{h}\right)  ^{\ast}\,\left[  1_{\left\{  k\geq0\right\}
}\left(  k\right)  \left(  \mathcal{O}\left(  h^{-1}\right)  +\mathcal{P\circ
O}\left(  h^{-1}\right)  \right)  \mathcal{F}_{\mathcal{V}}^{h}\varphi\right]
\right. \nonumber\\
& \nonumber\\
&  \left.  +1_{\left\{  x<b\right\}  }\left(  \mathcal{F}_{\mathcal{V}}%
^{h}\right)  ^{\ast}\,\left[  1_{\left\{  k<0\right\}  }\left(  k\right)
\left(  \mathcal{P\circ O}\left(  h^{-1}\right)  +\mathcal{O}\left(
h^{-1}\right)  \right)  \mathcal{F}_{\mathcal{V}}^{h}\varphi\right]  \right.
\,, \label{phi_formula}%
\end{align}
from which it follows: $\left.  \left\Vert \phi_{b}^{h}(\varphi,f)\right\Vert
_{L^{2}\left(  \mathbb{R}\right)  }\lesssim1/h\left\Vert \varphi\right\Vert
_{L^{2}\left(  \mathbb{R}\right)  }\right.  $. Recall that the generalized
Fourier transform $\mathcal{F}_{\mathcal{V}}^{h}$ is a bounded map from
$H^{2}\left(  \mathbb{R}\right)  $ to the weighted space $L^{2,2}\left(
\mathbb{R}\right)  $, defined by%
\begin{equation}
L^{2,\alpha}\left(  \mathbb{R}\right)  =\left\{  u\in L^{2}\left(
\mathbb{R}\right)  :\left(  1+k^{2}\right)  ^{\alpha/2}u\in L^{2}\left(
\mathbb{R}\right)  \right\}  \,,
\end{equation}
namely, we have%
\begin{equation}%
\begin{array}
[c]{ccc}%
\mathcal{F}_{\mathcal{V}}^{h}\in\mathcal{B}\left(  H^{2}\left(  \mathbb{R}%
\right)  ,L^{2,2}\left(  \mathbb{R}\right)  \right)  \,, & \text{ and} &
\left(  \mathcal{F}_{\mathcal{V}}^{h}\right)  ^{\ast}\in\mathcal{B}\left(
L^{2,2}\left(  \mathbb{R}\right)  ,H^{2}\left(  \mathbb{R}\right)  \right)
\,.
\end{array}
\label{gen_Fourier_rel}%
\end{equation}
Since $h\left(  \mathcal{O}\left(  h^{-1}\right)  +\mathcal{P\circ O}\left(
h^{-1}\right)  \right)  $ is bounded on $L^{2,2}\left(  \mathbb{R}\right)  $
uniformly w.r.t. $h\in\left(  0,h_{0}\right]  $, from (\ref{phi_formula}) we
also have: $\left.  \left\Vert \phi_{b}^{h}(\varphi,f)\right\Vert
_{H^{2}\left(  \mathbb{R}\backslash\left\{  b\right\}  \right)  }%
\lesssim1/h\left\Vert \varphi\right\Vert _{H^{2}\left(  \mathbb{R}\right)
}\right.  $. In the case of $\psi_{\alpha}^{h}(\varphi,g)$, the representation
(\ref{H_gen_eigenfun_k1}) allows similar computations leading to: $\left.
\left\Vert \psi_{b}^{h}(\varphi,g)\right\Vert _{L^{2}\left(  \mathbb{R}%
\right)  }\lesssim1/h^{2}\left\Vert \varphi\right\Vert _{L^{2}\left(
\mathbb{R}\right)  }\right.  $ and $\left.  \left\Vert \psi_{b}^{h}%
(\varphi,g)\right\Vert _{H^{2}\left(  \mathbb{R}\backslash\left\{  b\right\}
\right)  }\lesssim1/h^{2}\left\Vert \varphi\right\Vert _{H^{2}\left(
\mathbb{R}\right)  }\right.  $. For $\left.  \alpha=a\right.  $, a
representation of the type (\ref{phi_psi_formula}) for the maps
(\ref{Phi_alpha_h})-(\ref{Psi_alpha_h}) is obtained by a suitably adaptation
of the previous arguments.
\end{proof}

The stationary waves operators $\mathcal{W}_{\theta}^{h}$ are defined by the
integral kernel%
\begin{equation}
\mathcal{W}_{\theta}^{h}(x,y)=\int_{\mathbb{R}}\frac{dk}{2\pi h}\,\psi
_{\theta}^{h}(x,k)\left(  \psi_{0}^{h}(x,k)\right)  ^{\ast}\,.
\label{W_teta_ker_h}%
\end{equation}
These have been considered in a slightly different framework in \cite{Man2},
where, using an energy cutoff (corresponding to a cutoff in $k$ in
(\ref{W_teta_ker_h})) and suitable spectral assumptions, estimates of the type
(\ref{phi_psi_est}) are obtained and a small-$\theta$ expansion of
$\mathcal{W}_{\theta}^{h}$ is provided with for $h\in\left(  0,h_{0}\right]  $
(see \cite[Lemma 4.2 and Proposition 4.3]{Man2}). In our setting, the
positivity condition (\ref{V_pos1}) allows generalize this result as follows.

\begin{proposition}
\label{Proposition_W_cont}Let $h\in\left(  0,h_{0}\right]  $, with $h_{0}$
suitably small, $\mathcal{V}$ be defined according to (\ref{V_pos1}) and
$\left\vert \theta\right\vert \leq h^{N_{0}}$, with $N_{0}>2$. Then $\left\{
\mathcal{W}_{\theta}^{h}\,,\ \theta\in B_{h^{N_{0}}}(0)\right\}  $ form an
analytic family of bounded and invertible operators on $L^{2}(\mathbb{R})$
fulfilling the expansion%
\begin{equation}
\left\Vert \mathcal{W}_{\theta}^{h}-1_{L^{2}\left(  \mathbb{R}\right)
}\right\Vert _{\mathcal{L}\left(  L^{2}(\mathbb{R})\right)  }+\left\Vert
\left(  \mathcal{W}_{\theta}^{h}\right)  ^{-1}-1_{L^{2}\left(  \mathbb{R}%
\right)  }\right\Vert _{\mathcal{L}\left(  L^{2}(\mathbb{R})\right)  }\leq
C_{a,b,c}h^{N_{0}-2}\,, \label{W_teta_exp}%
\end{equation}
where $C_{a,b,c}$ is a positive constant depending on the data.

For each $\theta\in B_{h^{N_{0}}}(0)$, $\left.  \mathcal{W}_{\theta}^{h}%
\in\mathcal{B}\left(  H^{2}\left(  \mathbb{R}\right)  ,H^{2}\left(
\mathbb{R}\backslash\left\{  a,b\right\}  \right)  \right)  \right.  $ with
$\left.  ran\left(  \mathcal{W}_{\theta}^{h}\upharpoonright H^{2}\left(
\mathbb{R}\right)  \right)  =D\left(  \Delta_{\theta}\right)  \right.  $ and
it results%
\begin{equation}
H_{\theta}^{h}\mathcal{W}_{\theta}^{h}=\mathcal{W}_{\theta}^{h}H_{0}^{h}\,.
\label{W_teta_inter}%
\end{equation}
In particular, $\mathcal{W}_{\theta}^{h}$ is an isomorphism: $H^{2}\left(
\mathbb{R}\right)  \rightarrow D\left(  \Delta_{\theta}\right)  $ (considered
as an Hilbert subspace of $H^{2}\left(  \mathbb{R}\backslash\left\{
a,b\right\}  \right)  $).
\end{proposition}

\begin{proof}
Due to the assumptions: $\left\vert \theta\right\vert <h^{2}$, the formula
(\ref{gen_eigenfun_exp}) applies and the action of $\mathcal{W}_{\theta}^{h}$
on $\varphi\in L^{2}(\mathbb{R})$ writes as%
\begin{equation}
\mathcal{W}_{\theta}^{h}\varphi=\left(  \mathcal{F}_{\mathcal{V}}^{h}\right)
^{\ast}\left(  \mathcal{F}_{\mathcal{V}}^{h}\varphi\right)  (k)+\sum
_{\alpha=a,b}\int_{\mathbb{R}}\frac{dk}{\left(  2\pi h\right)  ^{1/2}%
}\,\left[  \mathcal{O}\left(  \frac{\theta k}{h}\right)  G^{\left\vert
k\right\vert ,h}\left(  \cdot,\alpha\right)  +\mathcal{O}\left(  \frac{\theta
k}{1+k}\right)  H^{\left\vert k\right\vert ,h}(\cdot,\alpha)\right]  \left(
\mathcal{F}_{\mathcal{V}}^{h}\varphi\right)  (k)\,.
\end{equation}
where $\mathcal{O}\left(  \cdot\right)  $ here denote bounded functions of the
variables $\left\{  k,\theta,h\right\}  $, holomorphic w.r.t. $\theta$. With
the notation introduced in (\ref{Phi_alpha_h})-(\ref{Psi_alpha_h}), the
identities: $\mathcal{O}\left(  g\right)  =g\,\mathcal{O}\left(  1\right)  $
(see the definition \ref{Landau_Notation}) and $1_{L^{2}\left(  \mathbb{R}%
\right)  }=\left(  \mathcal{F}_{\mathcal{V}}^{h}\right)  ^{\ast}\left(
\mathcal{F}_{\mathcal{V}}^{h}\varphi\right)  $ yield%
\begin{equation}
\left(  \mathcal{W}_{\theta}^{h}-1_{L^{2}\left(  \mathbb{R}\right)  }\right)
\varphi=\sum_{\alpha=a,b}\left[  \frac{\theta}{h}\,\phi_{\alpha}^{h}\left(
\varphi,\mathcal{O}\left(  k\right)  \right)  +\theta\,\psi_{\alpha}%
^{h}\left(  \varphi,\mathcal{O}\left(  k\left(  1+k\right)  ^{-1}\right)
\right)  \right]  \,, \label{W_teta_exp_h1}%
\end{equation}
Then, (\ref{phi_psi_est}) applies to the r.h.s. of (\ref{W_teta_exp_h1}) and
using $\left\vert \theta\right\vert <h^{N_{0}}$, we conclude that%
\begin{equation}
\left\Vert \mathcal{W}_{\theta}^{h}-1_{L^{2}\left(  \mathbb{R}\right)
}\right\Vert _{\mathcal{L}\left(  L^{2}(\mathbb{R})\right)  }\leq
C_{a,b,c}h^{N_{0}-2}\,. \label{W_teta_h_est}%
\end{equation}
Then, for $h_{0}$ suitably small, $\mathcal{W}_{\theta}^{h}\in\mathcal{B}%
\left(  L^{2}(\mathbb{R})\right)  $ is invertible and $\left(  \mathcal{W}%
_{\theta}^{h}\right)  ^{-1}$ fulfills an analogous estimate; this yields
(\ref{W_teta_exp}). The action of $\mathcal{W}_{\theta}^{h}$ over
$L^{2}\left(  \mathbb{R}\right)  $ is defined using the expansion
(\ref{W_teta_exp_h1}). Each of the maps $\phi_{\alpha}^{h}\left(
\cdot,\mathcal{O}\left(  k\right)  \right)  $ and $\psi_{\alpha}^{h}\left(
\cdot,\mathcal{O}\left(  k\left(  1+k\right)  ^{-1}\right)  \right)  $,
appearing in this formula expresses as a superposition of the form (cf.
(\ref{phi_psi_formula}))%
\begin{equation}
1_{\left\{  x\geq\alpha\right\}  }\mathcal{T}_{\mathcal{\alpha}}^{h}\left(
\mu_{1,\alpha}+\mathcal{P\circ}\mu_{2,\alpha}\right)  +1_{\left\{
x<\alpha\right\}  }\mathcal{T}_{\mathcal{\alpha}}^{h}\left(  \mu_{3,\alpha
}+\mathcal{P\circ}\mu_{4,\alpha}\right)
\end{equation}
$\mu_{i,\alpha}$ being, in our case, bounded functions of $\left\{
k,\theta,h\right\}  $, holomorphic w.r.t. $\theta$. Thus, $\phi_{\alpha}%
^{h}\left(  \cdot,\mathcal{O}\left(  k\right)  \right)  $ and \newline$\left.
\psi_{\alpha}^{h}\left(  \cdot,\mathcal{O}\left(  k\left(  1+k\right)
^{-1}\right)  \right)  \right.  $ define holomorphic families of bounded maps
on $L^{2}\left(  \mathbb{R}\right)  $ and, due to (\ref{W_teta_exp_h1}), this
still holds for $\mathcal{W}_{\theta}^{h}$.

Let us consider the action of $\mathcal{W}_{\theta}^{h}$ on $H^{2}\left(
\mathbb{R}\right)  $; using (\ref{W_teta_exp_h1}), (\ref{phi_psi_est_1}), we
get%
\begin{equation}
\left\Vert \mathcal{W}_{\theta}^{h}-1_{H^{2}\left(  \mathbb{R}\backslash
\left\{  a,b\right\}  \right)  }\right\Vert _{\mathcal{L}\left(  H^{2}\left(
\mathbb{R}\right)  ,H^{2}\left(  \mathbb{R}\backslash\left\{  a,b\right\}
\right)  \right)  }\leq C_{a,b,c}h^{N_{0}-2}\,. \label{W_teta_h_est_1}%
\end{equation}
This implies: $\left.  \mathcal{W}_{\theta}^{h}\in\mathcal{B}\left(
H^{2}\left(  \mathbb{R}\right)  ,H^{2}\left(  \mathbb{R}\backslash\left\{
a,b\right\}  \right)  \right)  \right.  $, while, from the definitions
(\ref{gen_eigenfun_eq}) and (\ref{W_teta_ker_h}), $\mathcal{W}_{\theta}%
^{h}\varphi$ fulfills the interface conditions (\ref{BC_theta}); it follows%
\begin{equation}
ran\left(  \mathcal{W}_{\theta}^{h}\upharpoonright H^{2}\left(  \mathbb{R}%
\right)  \right)  \subseteq D\left(  \Delta_{\theta}\right)  \,.
\label{W_teta_ran}%
\end{equation}
Let $\varphi\in H^{2}\left(  \mathbb{R}\right)  $; using the functional
calculus of $H_{0}^{h}$, we have: $\left.  \left(  \mathcal{F}_{\mathcal{V}%
}^{h}\left(  H_{0}^{h}\varphi\right)  \right)  (k)=k^{2}\left(  \mathcal{F}%
_{\mathcal{V}}^{h}\varphi\right)  (k)\right.  $, and, from the definition
(\ref{W_teta_ker_h}), the r.h.s. of (\ref{W_teta_inter}) writes as%
\begin{equation}
\mathcal{W}_{\theta}^{h}H_{0}^{h}\varphi=\int_{\mathbb{R}}\frac{dk}{2\pi
h}\,\psi_{\theta}^{h}(\cdot,k)k^{2}\left(  \mathcal{F}_{\mathcal{V}}%
^{h}\varphi\right)  (k)\,. \label{W_teta_inter1}%
\end{equation}
Using once more the definition (\ref{W_teta_ker_h}) and the relation: $\left.
\left(  H_{\theta}^{h}-k^{2}\right)  \psi_{\theta}^{h}(\cdot,k)=0\right.  $,
the l.h.s. of (\ref{W_teta_inter}) identifies with%
\begin{equation}
H_{\theta}^{h}\mathcal{W}_{\theta}^{h}\varphi=\int_{\mathbb{R}}\frac{dk}{2\pi
h}\,\psi_{\theta}^{h}(\cdot,k)k^{2}\left(  \mathcal{F}_{\mathcal{V}}%
^{h}\varphi\right)  (k)\,. \label{W_teta_inter2}%
\end{equation}
From (\ref{W_teta_inter1})-(\ref{W_teta_inter2}) we get the intertwining
relation (\ref{W_teta_inter}). Since $\left(  \mathcal{W}_{\theta}^{h}\right)
^{-1}$ exists, we also have: \newline$\left.  \left(  \mathcal{W}_{\theta}%
^{h}\right)  ^{-1}H_{\theta}^{h}=H_{0}^{h}\left(  \mathcal{W}_{\theta}%
^{h}\right)  ^{-1}\right.  $, which implies%
\begin{equation}
ran\left(  \left(  \mathcal{W}_{\theta}^{h}\right)  ^{-1}\upharpoonright
D\left(  \Delta_{\theta}\right)  \right)  \subseteq H^{2}\left(
\mathbb{R}\right)  \,. \label{W_teta_ran_1}%
\end{equation}
From (\ref{W_teta_ran}) and (\ref{W_teta_ran_1}), follows: $\left.  ran\left(
\mathcal{W}_{\theta}^{h}\upharpoonright H^{2}\left(  \mathbb{R}\right)
\right)  =D\left(  \Delta_{\theta}\right)  \right.  $. Thus, the restriction
$\mathcal{W}_{\theta}^{h}\upharpoonright H^{2}\left(  \mathbb{R}\right)  $ is
injective (since $\left(  \mathcal{W}_{\theta}^{h}\right)  ^{-1}$ exists in
$\mathcal{L}\left(  L^{2}(\mathbb{R})\right)  $), and surjective onto
$ran\left(  \mathcal{W}_{\theta}^{h}\upharpoonright H^{2}\left(
\mathbb{R}\right)  \right)  =D\left(  \Delta_{\theta}\right)  $.
\end{proof}

\begin{remark}
\label{Remark_W}The explicit bounds for the factors $\mathcal{O}\left(
\cdot\right)  $ appearing in (\ref{gen_eigenfun_exp}) depend on the trace
estimates provided by the Proposition \ref{Proposition_trace_est_h}. According
to the Remark \ref{Remark_trace}, these are independent from the potential's
profile, provided that the assumptions are fulfilled. As a consequence, the
expansion (\ref{gen_eigenfun_exp}), as well as the relations
(\ref{phi_psi_est})-(\ref{phi_psi_est_1}) and the expansion
(\ref{W_teta_h_est}) hold uniformly w.r.t. any family of potentials for which
the conditions (\ref{V_pos1}) hold for a fixed $c>0$.
\end{remark}

Due to the result of the Proposition \ref{Proposition_W_cont}, $\mathcal{W}%
_{\theta}^{h}$ is an invertible map as far as $h\in\left(  0,h_{0}\right]  $
and $\left\vert \theta\right\vert \leq h^{N_{0}}$ (with $N_{0}>2$ and $h_{0}$
small); under these conditions, the intertwining property (\ref{W_teta_inter})
yields a similarity between $H_{\theta}^{h}$ and $H_{0}^{h}$; this allows to
define the quantum dynamics generated by $H_{\theta}^{h}$ by conjugation.

\begin{proof}
[Proof of the Theorem \ref{Theorem_propagator}]The first part of the statement
follows from the Proposition \ref{Proposition_W_cont}.

For the second part, let us introduce%
\begin{equation}
e^{-itH_{\theta}^{h}}=\mathcal{W}_{\theta}^{h}e^{-itH_{0}^{h}}\left(
\mathcal{W}_{\theta}^{h}\right)  ^{-1}\,. \label{propagator_teta}%
\end{equation}
Under our assumptions on the potential, $iH_{0}^{h}$ generates a strongly
continuous group of unitary maps both on $L^{2}\left(  \mathbb{R}\right)  $
and on $H^{2}\left(  \mathbb{R}\right)  $. According to the results of the
Proposition \ref{Proposition_W_cont}, $\mathcal{W}_{\theta}^{h}$ is bounded
and invertible on $L^{2}\left(  \mathbb{R}\right)  $, while its restriction
$\mathcal{W}_{\theta}^{h}\upharpoonright H^{2}\left(  \mathbb{R}\right)
\in\mathcal{B}\left(  H^{2}\left(  \mathbb{R}\right)  ,H^{2}\left(
\mathbb{R}\backslash\left\{  a,b\right\}  \right)  \right)  $ has $ran\left(
\mathcal{W}_{\theta}^{h}\upharpoonright H^{2}\left(  \mathbb{R}\right)
\right)  =D\left(  \Delta_{\theta}\right)  $. Hence, $\mathcal{W}_{\theta}%
^{h}$ is a bijection: $H^{2}\left(  \mathbb{R}\right)  \rightarrow D\left(
\Delta_{\theta}\right)  $ and the modified propagator $e^{-itH_{\theta}^{h}}$
is strongly continuous both on $L^{2}\left(  \mathbb{R}\right)  $ and on
$D\left(  \Delta_{\theta}\right)  $ (w.r.t. the corresponding topologies).
Moreover, from the identity: $i\partial_{t}e^{-itH_{0}^{h}}\psi=H_{0}%
^{h}e^{-itH_{0}^{h}}\psi$, holding in $L^{2}\left(  \mathbb{R}\right)  $ for
any $\psi\in H^{2}\left(  \mathbb{R}\right)  $, it follows%
\begin{equation}
i\partial_{t}\left(  e^{-itH_{\theta}^{h}}u\right)  =H_{\theta}^{h}%
e^{-itH_{\theta}^{h}}u\,,\qquad u\in D\left(  \Delta_{\theta}\right)  \,.
\label{propagator_teta_eq}%
\end{equation}
Then $e^{-itH_{\theta}^{h}}$ identifies with the quantum dynamical system
generated by $iH_{\theta}^{h}$.

Finally, since $\mathcal{W}_{\theta}^{h}$ and $\left(  \mathcal{W}_{\theta
}^{h}\right)  ^{-1}$ are analytic w.r.t. $\theta$, $e^{-itH_{\theta}^{h}}$ has
the same regularity and the expansion (\ref{propagator_est_1}) follows from
(\ref{W_teta_exp}).
\end{proof}

\section{\label{Sec_Nonaut}The time dependent case}

We consider the time dependent family of modified operators $H_{\theta}%
^{h}\left(  t\right)  $ defined according to (\ref{H_teta_h_t}) when the
potential is a continuous function of the time fulfilling the conditions%
\begin{equation}%
\begin{array}
[c]{ccccc}%
\mathcal{V}\left(  t\right)  \in\mathcal{C}^{0}\left(  \left[  0,T\right]
,L^{\infty}(\mathbb{R},\mathbb{R})\right)  \,, &  & \text{supp\thinspace
}\mathcal{V}\left(  t\right)  =\left[  a,b\right]  \,, &  & 1_{\left[
a,b\right]  }\mathcal{V}\left(  t\right)  >c\,,
\end{array}
\label{V_pos_t}%
\end{equation}
for a suitable $c>0$. The $\theta$-dependent time propagator $U_{\theta}%
^{h}\left(  t,s\right)  $ associated to $H_{\theta}^{h}\left(  t\right)  $
solves the evolution problem%
\begin{equation}
\left\{
\begin{array}
[c]{l}%
i\partial_{t}U_{\theta}^{h}\left(  t,s\right)  u=H_{\theta}^{h}\left(
t\right)  U_{\theta}^{h}\left(  t,s\right)  u\,,\\
\\
U_{\theta}^{h}\left(  s,s\right)  u=u\,,\quad u\in D\left(  \Delta_{\theta
}\right)  \,,\quad0\leq s\leq t\leq T\,.
\end{array}
\right.  \label{Time_dependent_eq}%
\end{equation}

A standard strategy in the definition of the quantum dynamical system
generated by a non-autonomous Hamiltonian, consists in using an approximating
sequence whose terms are stepwise products of propagators associated to the
'instantaneous' Hamiltonians (cf. \cite{Yosh}). This approach requires
stability estimates for the product of the instantaneous propagators in
suitable spaces.

\subsection{Stability estimates}

In the following, $D\left(  \Delta_{\theta}\right)  $ is considered as an
Hilbert subspace of $H^{2}\left(  \mathbb{R}\backslash\left\{  a,b\right\}
\right)  $ (see the definition (\ref{BC_theta})) and the notation
$\mathcal{L}\left(  D\left(  \Delta_{\theta}\right)  \right)  $ refers to the
linear operators on $D\left(  \Delta_{\theta}\right)  $ w.r.t. its topology.

In the time dependent case, the instantaneous propagators $e^{-itH_{\theta
}^{h}\left(  s\right)  }$ verify the relations%
\begin{equation}
e^{-itH_{\theta}^{h}\left(  s\right)  }=\mathcal{W}_{\theta}^{h}\left(
s\right)  e^{-itH_{0}^{h}\left(  s\right)  }\left(  \mathcal{W}_{\theta}%
^{h}\left(  s\right)  \right)  ^{-1}\,, \label{propagator_teta_t}%
\end{equation}
where the maps $\mathcal{W}_{\theta}^{h}\left(  t\right)  $ now depend on time
according to $\mathcal{V}\left(  t\right)  $. Let us recall from results of
the Proposition \ref{Proposition_W_cont} that, for any $t\in\left[
0,T\right]  $, $\mathcal{W}_{\theta}^{h}\left(  t\right)  $ is bounded and
invertible on $L^{2}\left(  \mathbb{R}\right)  $, while its restriction to
$H^{2}\left(  \mathbb{R}\right)  $ is a bijection: $H^{2}\left(
\mathbb{R}\right)  \rightarrow D\left(  \Delta_{\theta}\right)  $. In
particular, following the Remark \ref{Remark_W}, from the estimate
(\ref{W_teta_h_est}) we get%
\begin{equation}
\sup_{\substack{t\in\left[  0,T\right]  \\h\in\left(  0,h_{0}\right]
}}\left\{  \left\Vert \mathcal{W}_{\theta}^{h}\left(  t\right)  \right\Vert
_{\mathcal{L}\left(  L^{2}(\mathbb{R})\right)  }+\left\Vert \left(
\mathcal{W}_{\theta}^{h}\left(  t\right)  \right)  ^{-1}\right\Vert
_{\mathcal{L}\left(  L^{2}(\mathbb{R})\right)  }\right\}  \leq A_{a,b,c}\,,
\label{W_teta_time_est}%
\end{equation}
for some $A_{a,b,c}>0$ depending on the data, provided that $h_{0}$ is
suitably small and $\left\vert \theta\right\vert \leq h^{N_{0}}$, with
$N_{0}>2$. The selfadjointness of $H_{0}^{h}\left(  s\right)  $ then yields%
\begin{equation}
\sup_{\substack{t\in\mathbb{R},\ s\in\left[  0,T\right]  \\h\in\left(
0,h_{0}\right]  }}\left\Vert e^{-itH_{\theta}^{h}\left(  s\right)
}\right\Vert _{\mathcal{L}\left(  L^{2}(\mathbb{R})\right)  }\leq
K_{a,b,c}\,,\qquad\forall\,\left\vert \theta\right\vert \leq h^{N_{0}%
}\,,\ N_{0}>2\,, \label{propagator_bound_0}%
\end{equation}
for a suitable $K_{a,b,c}>0$. Under the same assumptions, the estimate
(\ref{W_teta_h_est_1}) and the Remark \ref{Remark_W} suggest%
\begin{equation}
\sup_{\substack{t\in\left[  0,T\right]  \\h\in\left(  0,h_{0}\right]
}}\left\Vert \mathcal{W}_{\theta}^{h}\left(  t\right)  \right\Vert
_{\mathcal{L}\left(  H^{2}\left(  \mathbb{R}\right)  ,H^{2}\left(
\mathbb{R}\backslash\left\{  a,b\right\}  \right)  \right)  }\leq B_{a,b,c}\,.
\label{W_teta_time_est_1}%
\end{equation}
According to the results of the Proposition \ref{Proposition_W_cont},
$ran\left(  \mathcal{W}_{\theta}^{h}\left(  t\right)  \upharpoonright
H^{2}\left(  \mathbb{R}\right)  \right)  =D\left(  \Delta_{\theta}\right)  $
for any $t$ and $\left(  \mathcal{W}_{\theta}^{h}\left(  t\right)  \right)
^{-1}$ exists in $\mathcal{B}\left(  D\left(  \Delta_{\theta}\right)
,H^{2}\left(  \mathbb{R}\right)  \right)  $. Then (\ref{W_teta_time_est_1})
yields%
\begin{equation}
\sup_{\substack{t\in\left[  0,T\right]  \\h\in\left(  0,h_{0}\right]
}}\left\Vert \left(  \mathcal{W}_{\theta}^{h}\left(  t\right)  \right)
^{-1}\right\Vert _{\mathcal{L}\left(  D\left(  \Delta_{\theta}\right)
,H^{2}\left(  \mathbb{R}\right)  \right)  }\leq C_{a,b,c}\,.
\label{W_teta_time_est_2}%
\end{equation}
Since $e^{-itH_{0}^{h}\left(  s\right)  }$ defines a unitary (strongly
continuous) flow on $H^{2}\left(  \mathbb{R}\right)  $, from the definition
(\ref{propagator_teta_t}) we obtain%
\begin{equation}
\sup_{\substack{t\in\mathbb{R},\ s\in\left[  0,T\right]  \\h\in\left(
0,h_{0}\right]  }}\left\Vert e^{-itH_{\theta}^{h}\left(  s\right)
}\right\Vert _{\mathcal{L}\left(  D\left(  \Delta_{\theta}\right)  \right)
}\leq\tilde{K}_{a,b,c}\,,\qquad\forall\,\left\vert \theta\right\vert \leq
h^{N_{0}}\,,\ N_{0}>2\,. \label{propagator_bound_0_reg}%
\end{equation}

For $T>0$ we next introduce the partition $\left[  0,T\right]  =\cup_{j=1}%
^{n}\left[  t_{j-1},t_{j}\right]  $, where $t_{j}=jT/n$ and $t_{0}=0$; the
step-propagators $U_{\theta,n}^{h}\left(  t,s\right)  $ are defined by%
\begin{equation}
U_{\theta,n}^{h}\left(  t,s\right)  =\left\{
\begin{array}
[c]{lll}%
e^{-i\left(  t-s\right)  H_{\theta}^{h}\left(  t_{j}\right)  }\,, &  &
s,t\in\left[  t_{j-1},t_{j}\right]  \,,\\
&  & \\
U_{\theta,n}^{h}\left(  t,t_{k+j-1}\right)  U_{\theta,n}^{h}\left(
t_{k+j-1},t_{k+j-2}\right)  \cdot\cdot\cdot U_{\theta,n}^{h}\left(
t_{j},s\right)  \,, &  & s\in\left[  t_{j-1},t_{j}\right]  \,,\ t\in\left[
t_{k+j-1},t_{k+j}\right]  \,.
\end{array}
\right.  \label{Propagator_n}%
\end{equation}
Under the assumption of the Theorem \ref{Theorem_propagator}, each factor in
$U_{\theta,n}^{h}\left(  t,s\right)  $ defines a $\theta$-holomorphic family
of bounded operators on $L^{2}\left(  \mathbb{R}\right)  $, strongly
continuous w.r.t. the time variables. Then for each $n$, $U_{\theta,n}%
^{h}\left(  t,s\right)  $ is $\theta$-holomorphic and strongly continuous in
$t$ and $s$ on $L^{2}\left(  \mathbb{R}\right)  $, while, according to its
definition, we have%
\begin{equation}
U_{\theta,n}^{h}\left(  s,s\right)  =1_{L^{2}\left(  \mathbb{R}\right)
}\,,\quad U_{\theta,n}^{h}\left(  t,s\right)  =U_{\theta,n}^{h}\left(
t,r\right)  U_{\theta,n}^{h}\left(  r,s\right)  \,,\quad\forall\,s\leq r\leq
t\,. \label{propagator_id}%
\end{equation}
The result in Theorem \ref{Theorem_propagator} also imply that each factor in
$U_{\theta,n}^{h}\left(  t,s\right)  $ is bounded on $D\left(  \Delta_{\theta
}\right)  $ and strongly continuous in the time variables (w.r.t. the
$\mathcal{L}\left(  D\left(  \Delta_{\theta}\right)  \right)  $ topology).
Thus, $U_{\theta,n}^{h}\left(  t,s\right)  $ is bounded strongly continuous in
$t$ and $s$ on $D\left(  \Delta_{\theta}\right)  $ and introducing:
$H_{\theta,n}^{h}\left(  t\right)  =H_{\theta}^{h}\left(  \frac{T}{n}\left[
\frac{nt}{T}\right]  \right)  $ ($\left[  \cdot\right]  $ denotes the floor
function), from (\ref{propagator_teta_eq}) the identity%
\begin{equation}
i\partial_{t}U_{\theta,n}^{h}\left(  t,s\right)  u=H_{\theta,n}^{h}\left(
t\right)  U_{\theta,n}^{h}\left(  t,s\right)  u\,,\quad\forall\,0\leq s\leq
t\leq T\,,\ u\in D\left(  \Delta_{\theta}\right)  \,, \label{Schrodinger_n}%
\end{equation}
holds in $L^{2}\left(  \mathbb{R}\right)  $. The additional condition (see the
definition (\ref{Multipliers}))%
\begin{equation}%
\begin{array}
[c]{ccc}%
\mathcal{V}\left(  t\right)  -\mathcal{V}\left(  s\right)  \in W_{0}%
^{2,\infty}\left(  \left[  a,b\right]  \right)  \,, &  & \forall
\,t,s\in\left[  T,0\right]  \,.
\end{array}
\label{V_cond}%
\end{equation}
is next used to obtain stability estimates for the sequence $U_{\theta,n}%
^{h}\left(  t,s\right)  $.

\begin{lemma}
\label{Lemma_propagator_n}Let $\mathcal{V}\left(  t\right)  $ fulfills the
conditions (\ref{V_pos_t}), $h\in\left(  0,h_{0}\right]  $, with $h_{0}$
suitably small, and $\left\vert \theta\right\vert \leq h^{N_{0}}$, with
$N_{0}>2$. There exist $C_{a,b,c}$ and $\tilde{C}_{a,b,c}$ positive and
possibly depending on the data, such that
\begin{equation}
\sup_{\substack{t,s\in\left[  0,T\right]  \\n\in\mathbb{N}^{\ast}%
,\ h\in\left(  0,h_{0}\right]  }}\left\Vert U_{\theta,n}^{h}\left(
t,s\right)  \right\Vert _{\mathcal{L}\left(  L^{2}\left(  \mathbb{R}\right)
\right)  }\leq\exp\left(  C_{a,b,c}\sup_{t\in\left[  0,T\right]  }\left\Vert
\mathcal{V}\left(  t\right)  \right\Vert _{L^{\infty}\left(  \mathbb{R}%
\right)  }\right)  \,,\qquad t\geq s\,. \label{propagator_bound_1}%
\end{equation}
If in addition (\ref{V_cond}) holds, then%
\begin{equation}
\sup_{\substack{t,s\in\left[  0,T\right]  \\n\in\mathbb{N}^{\ast}%
,\ h\in\left(  0,h_{0}\right]  }}\left\Vert U_{\theta,n}^{h}\left(
t,s\right)  \right\Vert _{\mathcal{L}\left(  D\left(  \Delta_{\theta}\right)
\right)  }\leq\exp\left(  \tilde{C}_{a,b,c}\sup_{t,s\in\left[  0,T\right]
}\left\Vert \mathcal{V}\left(  t\right)  -\mathcal{V}\left(  s\right)
\right\Vert _{W^{2,\infty}\left(  \left[  a,b\right]  \right)  }\right)
\,,\qquad t\geq s\,. \label{propagator_bound_reg_1}%
\end{equation}

\end{lemma}

\begin{proof}
Let us fix $j\in\left\{  1,...,n\right\}  $, $s\in\left[  t_{j-1}%
,t_{j}\right]  $ and consider $U_{\theta,n}^{h}\left(  t,s\right)  $; if
$\ t\in\left[  t_{j-1},t_{j}\right]  $, then (\ref{propagator_bound_1}) is a
consequence of (\ref{propagator_bound_0}). If $t\in\left[  t_{k+j-1}%
,t_{k+j}\right]  $ for some $k\in\left\{  1,...,n-j\right\}  $, then
$U_{\theta,n}^{h}\left(  t,s\right)  $ is a product of $k+1$ terms and writes
as%
\begin{equation}
U_{\theta,n}^{h}\left(  t,s\right)  =e^{-i\left(  t-t_{k+j-1}\right)
H_{\theta}^{h}\left(  t_{k+j-1}\right)  }\left(
{\displaystyle\prod\limits_{\ell=1}^{k-1}}
U_{\theta,n}^{h}\left(  t_{k+j-\ell},t_{k+j-\ell-1}\right)  \right)
e^{-i\left(  t_{j}-s\right)  H_{\theta}^{h}\left(  t_{j-1}\right)  }\,.
\label{propagator_fact_0}%
\end{equation}
Let $m\in\left\{  1,...,n\right\}  $ and%
\begin{equation}
I_{\theta}^{h}\left(  \tau,r,m\right)  =-i%
{\displaystyle\int\limits_{r}^{\tau}}
e^{-i\left(  \tau-x\right)  H_{\theta}^{h}\left(  t_{j-1}\right)  }\left(
\mathcal{V}\left(  t_{m-1}\right)  -\mathcal{V}\left(  t_{j-1}\right)
\right)  e^{-i\left(  x-r\right)  H_{\theta}^{h}\left(  t_{m-1}\right)
}\,dx\,. \label{propagator_rem}%
\end{equation}
Each factor in $U_{\theta,n}^{h}\left(  t,s\right)  $ allows the
representation%
\begin{equation}
U_{\theta,n}^{h}\left(  \tau,r\right)  =e^{-i\left(  \tau-r\right)  H_{\theta
}^{h}\left(  t_{m-1}\right)  }=e^{-i\left(  \tau-r\right)  H_{\theta}%
^{h}\left(  t_{j-1}\right)  }+I_{\theta}^{h}\left(  \tau,r,m\right)
\,,\qquad\tau,r\in\left[  t_{m-1},t_{m}\right]  \,, \label{propagator_fact_1}%
\end{equation}
and the identity (\ref{propagator_fact_0}) rephrases as%
\begin{align}
&  \left.  U_{\theta,n}^{h}\left(  t,s\right)  =\left(  e^{-i\left(
t-t_{k+j-1}\right)  H_{\theta}^{h}\left(  t_{j-1}\right)  }+I_{\theta}%
^{h}\left(  t,t_{k+j-1},k+j\right)  \right)  \circ\right. \nonumber\\
& \nonumber\\
&  \left.  \left(
{\displaystyle\prod\limits_{\ell=1}^{k-1}}
\left(  e^{-i\left(  t_{k+j-\ell}-t_{k+j-\ell-1}\right)  H_{\theta}^{h}\left(
t_{j-1}\right)  }+I_{\theta}^{h}\left(  t_{k+j-\ell},t_{k+j-\ell-1}%
,k+j-\ell\right)  \right)  \right)  \circ\right.  \qquad\qquad\qquad
\qquad\nonumber\\
& \nonumber\\
&  \left.  \left(  e^{-i\left(  t_{j}-s\right)  H_{\theta}^{h}\left(
t_{j-1}\right)  }+I_{\theta}^{h}\left(  t_{j},s,j\right)  \right)  \right.
\,.
\end{align}
To simplify the notation let us fix $n_{0}=k+j\leq n$ and assume, without loss
of generality, assume that $j=1$ (which implies $t_{j-1}=t_{0}=0$); it follows%
\begin{align}
&  \left.  U_{\theta,n}^{h}\left(  t,s\right)  =\left(  e^{-i\left(
t-t_{n_{0}-1}\right)  H_{\theta}^{h}\left(  0\right)  }+I_{\theta}^{h}\left(
t,t_{n_{0}-1},n_{0}\right)  \right)  \circ\right. \nonumber\\
& \nonumber\\
&  \left.  \left(
{\displaystyle\prod\limits_{\ell=1}^{n_{0}-2}}
\left(  e^{-i\left(  t_{n_{0}-\ell}-t_{n_{0}-\ell-1}\right)  H_{\theta}%
^{h}\left(  0\right)  }+I_{\theta}^{h}\left(  t_{n_{0}-\ell},t_{n_{0}-\ell
-1},n_{0}-\ell\right)  \right)  \right)  \left(  e^{-i\left(  t_{1}-s\right)
H_{\theta}^{h}\left(  0\right)  }+I_{\theta}^{h}\left(  t_{1},s,1\right)
\right)  \right.  \,.
\end{align}
Each contribution to the sum obtained by expanding this product of $n_{0}$
binomials is a product of $n_{0}$ factors corresponding either to the quantum
propagators associated to $H_{\theta}^{h}\left(  0\right)  $ either to the
operators $I_{\theta}^{h}$. Recalling that in $\left(
\begin{array}
[c]{c}%
n_{0}\\
m
\end{array}
\right)  $ terms of this sum the factors $I_{\theta}^{h}$ appear $m$ times, we
get%
\begin{equation}
U_{\theta,n}^{h}\left(  t,s\right)  =\sum_{m=0}^{n_{0}}\sum_{\ell=1}^{b_{m}%
}F_{\theta,m}^{h}\left(  t,s,\ell\right)  \,,\quad b_{m}=\left(
\begin{array}
[c]{c}%
n_{0}\\
m
\end{array}
\right)  \,, \label{propagator_fact_2}%
\end{equation}
where $F_{\theta}^{h}\left(  t,s,\ell\right)  $, possibly depending on $t$ and
$s$, denote the contributions to $U_{\theta,n}^{h}\left(  t,s\right)  $ where
$m$ terms of the type $I_{\theta}^{h}$ appear; using the group properties of
$e^{-i\tau H_{\theta}^{h}\left(  0\right)  }$ and the definition
(\ref{propagator_rem}), these are factorized according to
\begin{equation}
F_{\theta,m}^{h}\left(  t,s,\ell\right)  =\left[
{\displaystyle\prod\limits_{p=1}^{m}}
\left(  -i%
{\displaystyle\int\limits_{t_{j_{p}-1}}^{t_{j_{p}}}}
e^{-i\left(  \tau_{p}-x\right)  H_{\theta}^{h}\left(  0\right)  }\left(
\mathcal{V}\left(  t_{j_{p}-1}\right)  -\mathcal{V}\left(  0\right)  \right)
e^{-i\left(  x-r_{p}\right)  H_{\theta}^{h}\left(  t_{j_{p}-1}\right)
}\,dx\right)  \right]  e^{-i\left(  z_{p}-s\right)  H_{\theta}^{h}\left(
0\right)  }\,, \label{propagator_fact_3}%
\end{equation}
being $j_{p}$ is a strictly decreasing subsequence of $\left\{  1,2,...,n_{0}%
\right\}  $ and $\tau_{p}\geq r_{p}$, $z_{p}\geq s$ suitable values in
$\left[  0,T\right]  $ depending on $t,s$ and $\ell$. The estimate
(\ref{propagator_bound_0}) and%
\begin{equation}
\left\Vert
{\displaystyle\int\limits_{t_{j_{p}-1}}^{t_{j_{p}}}}
e^{-i\left(  \tau_{p}-x\right)  H_{\theta}^{h}\left(  0\right)  }\left(
\mathcal{V}\left(  t_{j_{p}-1}\right)  -\mathcal{V}\left(  0\right)  \right)
e^{-i\left(  x-r_{p}\right)  H_{\theta}^{h}\left(  t_{j_{p}-1}\right)
}\,dx\right\Vert _{\mathcal{L}\left(  L^{2}\left(  \mathbb{R}\right)  \right)
}\leq2\frac{K_{a,b,c}^{2}}{n}\sup_{t\in\left[  0,T\right]  }\left\Vert
\mathcal{V}\left(  t\right)  \right\Vert _{L^{\infty}\left(  \mathbb{R}%
\right)  }\,, \label{propagator_rem_bound}%
\end{equation}
imply%
\begin{equation}
\left\Vert F_{\theta,m}^{h}\left(  t,s,\ell\right)  \right\Vert _{\mathcal{L}%
\left(  L^{2}\left(  \mathbb{R}\right)  \right)  }\leq\left(  \frac{C_{a,b,c}%
}{n}\sup_{t\in\left[  0,T\right]  }\left\Vert \mathcal{V}\left(  t\right)
\right\Vert _{L^{\infty}\left(  \mathbb{R}\right)  }\right)  ^{m}\,,
\end{equation}
for some $C_{a,b,c}>0$ depending on the data; setting: $C_{a,b,c,\mathcal{V}%
}=C_{a,b,c}\sup_{t\in\left[  0,T\right]  }\left\Vert \mathcal{V}\left(
t\right)  \right\Vert _{L^{\infty}\left(  \mathbb{R}\right)  }$, from
(\ref{propagator_fact_2}) follows%
\begin{equation}
\left\Vert U_{\theta,n}^{h}\left(  t,s\right)  \right\Vert _{\mathcal{L}%
\left(  L^{2}\left(  \mathbb{R}\right)  \right)  }\leq\sum_{m=0}^{n_{0}%
}\left(
\begin{array}
[c]{c}%
n_{0}\\
m
\end{array}
\right)  \left(  \frac{C_{a,b,c,\mathcal{V}}}{n}\right)  ^{m}\leq\sum
_{m=0}^{n}\left(
\begin{array}
[c]{c}%
n\\
m
\end{array}
\right)  \left(  \frac{C_{a,b,c,\mathcal{V}}}{n}\right)  ^{m}=\left(
1+\frac{C_{a,b,c,\mathcal{V}}}{n}\right)  ^{n}\,.
\end{equation}
Since this bound holds independently from $t,s\in\left[  0,T\right]  $ and
$h\in\left(  0,h_{0}\right]  $ provided that $\left\vert \theta\right\vert
\leq h^{N_{0}}$, with $N_{0}>2$, we get%
\begin{equation}
\sup_{\substack{t,s\in\left[  0,T\right]  \\h\in\left(  0,h_{0}\right]
}}\left\Vert U_{\theta,n}^{h}\left(  t,s\right)  \right\Vert _{\mathcal{L}%
\left(  L^{2}\left(  \mathbb{R}\right)  \right)  }\leq\left(  1+\frac
{C_{a,b,c,\mathcal{V}}}{n}\right)  ^{n}\,,\qquad t\geq s\,.
\end{equation}
Then, the uniform estimate (\ref{propagator_bound_1}) follows from the limit:
$\lim_{n\rightarrow\infty}\left(  1+\frac{C_{a,b,c,\mathcal{V}}}{n}\right)
^{n}=e^{C_{a,b,c,\mathcal{V}}}$.

Let us remark that a function $\psi\in L^{\infty}\left(  \mathbb{R}\right)
\cap W_{0}^{2,\infty}\left(  \left[  a,b\right]  \right)  $ such that:
supp\thinspace$\psi=\left[  a,b\right]  $ is a multiplier of $D\left(
\Delta_{\theta}\right)  $; in particular, for any $u\in D\left(
\Delta_{\theta}\right)  $ it results: $\psi u\in H_{0}^{2}\left(
\mathbb{R}\backslash\left\{  a,b\right\}  \right)  \subset D\left(
\Delta_{\theta}\right)  $ and $\left\Vert \psi u\right\Vert _{H^{2}\left(
\mathbb{R}\backslash\left\{  a,b\right\}  \right)  }\leq\left\Vert
\psi\right\Vert _{W^{2,\infty}\left(  \left[  a,b\right]  \right)  }\left\Vert
u\right\Vert _{H^{2}\left(  \mathbb{R}\backslash\left\{  a,b\right\}  \right)
}$. Then, the estimate (\ref{propagator_bound_0_reg}) and the assumption
(\ref{V_cond}) yield%
\begin{equation}
\left\Vert
{\displaystyle\int\limits_{t_{j_{p}-1}}^{t_{j_{p}}}}
e^{-i\left(  \tau_{p}-x\right)  H_{\theta}^{h}\left(  0\right)  }\left(
\mathcal{V}\left(  t_{j_{p}-1}\right)  -\mathcal{V}\left(  0\right)  \right)
e^{-i\left(  x-r_{p}\right)  H_{\theta}^{h}\left(  t_{j_{p}-1}\right)
}\,dx\right\Vert _{\mathcal{L}\left(  D\left(  \Delta_{\theta}\right)
\right)  }\leq\frac{\tilde{K}_{a,b,c}^{2}}{n}\sup_{t,s\in\left[  0,T\right]
}\left\Vert \mathcal{V}\left(  t\right)  -\mathcal{V}\left(  s\right)
\right\Vert _{W^{2,\infty}\left(  \left[  a,b\right]  \right)  }\,,
\end{equation}
Proceeding as before we have%
\begin{equation}
\left\Vert F_{\theta}^{h}\left(  t,s,\ell\right)  \right\Vert _{\mathcal{L}%
\left(  D\left(  \Delta_{\theta}\right)  \right)  }\leq\left(  \frac{\tilde
{C}_{a,b,c}}{n}\sup_{t,s\in\left[  0,T\right]  }\left\Vert \mathcal{V}\left(
t\right)  -\mathcal{V}\left(  s\right)  \right\Vert _{W^{2,\infty}\left(
\left[  a,b\right]  \right)  }\right)  ^{m}\,,
\end{equation}
for some $\tilde{C}_{a,b,c}>0$ depending on the data, and the representation
(\ref{propagator_fact_2})-(\ref{propagator_fact_3}) lead us to
(\ref{propagator_bound_reg_1})
\end{proof}

\begin{remark}
\label{Remark_propagator_n}The constants $C_{a,b,c}$ and $\tilde{C}_{a,b,c}$
in (\ref{propagator_bound_1})-(\ref{propagator_bound_reg_1}) do not depend on
$T$ once the assumptions (\ref{V_pos_t}), (\ref{V_cond}) are fulfilled.
\end{remark}

\subsection{The existence of the dynamics}

We next show that $U_{\theta,n}^{h}\left(  t,s\right)  $ approximates the the
dynamical system $U_{\theta}^{h}\left(  t,s\right)  $ introduced in
(\ref{Time_dependent_eq}); the proof adapts the strategy used in
\cite[Theorems 4.1 and 5.1]{Kato1} to our framework.

\begin{proposition}
\label{Proposition_propagator_lim}Under the assumptions of the Lemma
\ref{Lemma_propagator_n}, the sequence $U_{\theta,n}^{h}\left(  t,s\right)  $
uniformly converges in the $\mathcal{L}\left(  L^{2}\left(  \mathbb{R}\right)
\right)  $ topology to a limit operator $U_{\theta}^{h}\left(  t,s\right)  $
such that%
\begin{equation}
\sup_{\substack{t,s\in\left[  0,T\right]  \\h\in\left(  0,h_{0}\right]
}}\left\Vert U_{\theta}^{h}\left(  t,s\right)  \right\Vert _{\mathcal{L}%
\left(  L^{2}\left(  \mathbb{R}\right)  \right)  }\leq\exp\left(
C_{a,b,c}\sup_{t\in\left[  0,T\right]  }\left\Vert \mathcal{V}\left(
t\right)  \right\Vert _{L^{\infty}\left(  \mathbb{R}\right)  }\right)
\,,\qquad t\geq s\,, \label{propagator_bound_lim}%
\end{equation}
Moreover $U_{\theta}^{h}\left(  t,s\right)  \in\mathcal{L}\left(  D\left(
\Delta_{\theta}\right)  \right)  $ with
\begin{equation}
\sup_{\substack{t,s\in\left[  0,T\right]  \\h\in\left(  0,h_{0}\right]
}}\left\Vert U_{\theta}^{h}\left(  t,s\right)  \right\Vert _{\mathcal{L}%
\left(  D\left(  \Delta_{\theta}\right)  \right)  }\leq\exp\left(  \tilde
{C}_{a,b,c}\sup_{t,s\in\left[  0,T\right]  }\left\Vert \mathcal{V}\left(
t\right)  -\mathcal{V}\left(  s\right)  \right\Vert _{W^{2,\infty}\left(
\left[  a,b\right]  \right)  }\right)  \qquad t\geq s\,.
\label{propagator_bound_lim_reg}%
\end{equation}
The positive constants $C_{a,b,c}$, $\tilde{C}_{a,b,c}$, possibly depending on
the data, are independent from $T$.
\end{proposition}

\begin{proof}
Let $s,t\in\left[  0,T\right]  $ and $t\geq s$; from (\ref{Schrodinger_n}),
the relation%
\begin{equation}
\left(  U_{\theta,n}^{h}\left(  t,s\right)  -U_{\theta,m}^{h}\left(
t,s\right)  \right)  u=-i\int_{s}^{t}U_{\theta,n}^{h}\left(  t,t^{\prime
}\right)  \left(  H_{\theta,n}^{h}\left(  t^{\prime}\right)  -H_{\theta,m}%
^{h}\left(  t^{\prime}\right)  \right)  U_{\theta,m}^{h}\left(  t^{\prime
},s\right)  u\,dt^{\prime}\,, \label{Cauchy_id_0}%
\end{equation}
holds for any $u\in D\left(  \Delta_{\theta}\right)  $. The difference at the
r.h.s. writes as%
\begin{equation}
H_{\theta,n}^{h}\left(  t^{\prime}\right)  -H_{\theta,m}^{h}\left(  t^{\prime
}\right)  =\mathcal{V}\left(  \frac{T}{n}\left[  \frac{nt^{\prime}}{T}\right]
\right)  -\mathcal{V}\left(  \frac{T}{m}\left[  \frac{mt^{\prime}}{T}\right]
\right)  \,, \label{Cauchy_cond2}%
\end{equation}
and (\ref{Cauchy_id_0}) rephrases as%
\begin{equation}
\left(  U_{\theta,n}^{h}\left(  t,s\right)  -U_{\theta,m}^{h}\left(
t,s\right)  \right)  u=-i\int_{s}^{t}U_{\theta,n}^{h}\left(  t,t^{\prime
}\right)  \left(  \mathcal{V}\left(  \frac{T}{n}\left[  \frac{nt^{\prime}}%
{T}\right]  \right)  -\mathcal{V}\left(  \frac{T}{m}\left[  \frac{mt^{\prime}%
}{T}\right]  \right)  \right)  U_{\theta,m}^{h}\left(  t^{\prime},s\right)
u\,dt^{\prime}\,. \label{Cauchy_id}%
\end{equation}
Since both the l.h.s and at the r.h.s. of (\ref{Cauchy_id}) define bounded
operators on $L^{2}\left(  \mathbb{R}\right)  $, the density of the inclusion
$D\left(  \Delta_{\theta}\right)  \subset L^{2}\left(  \mathbb{R}\right)  $
allows to extend this identity to the whole space. Using the result of the
Lemma \ref{Lemma_propagator_n}, we get the estimate%
\begin{equation}
\left\Vert \left(  U_{\theta,n}^{h}\left(  t,s\right)  -U_{\theta,m}%
^{h}\left(  t,s\right)  \right)  u\right\Vert _{L^{2}\left(  \mathbb{R}%
\right)  }\leq M_{a,b,c}^{2}%
{\displaystyle\int_{0}^{T}}
\left\Vert \mathcal{V}\left(  \frac{T}{n}\left[  \frac{nt^{\prime}}{T}\right]
\right)  -\mathcal{V}\left(  \frac{T}{m}\left[  \frac{mt^{\prime}}{T}\right]
\right)  \right\Vert _{L^{\infty}\left(  \mathbb{R}\right)  }\left\Vert
u\right\Vert _{L^{2}\left(  \mathbb{R}\right)  }\,dt^{\prime}\,,
\label{Cauchy_cond_1}%
\end{equation}
while the regularity of $\mathcal{V}\left(  t\right)  $ yields%
\begin{equation}
\lim_{n,m\rightarrow\infty}\left\Vert \mathcal{V}\left(  \frac{T}{n}\left[
\frac{nt^{\prime}}{T}\right]  \right)  -\mathcal{V}\left(  \frac{T}{m}\left[
\frac{mt^{\prime}}{T}\right]  \right)  \right\Vert _{L^{\infty}\left(
\mathbb{R}\right)  }=0\,. \label{Cauchy_cond_2}%
\end{equation}
Hence, for any $u\in L^{2}\left(  \mathbb{R}\right)  $, $U_{\theta,n}%
^{h}\left(  t,s\right)  u$ forms a Cauchy sequence in $L^{2}\left(
\mathbb{R}\right)  $, uniformly w.r.t. $t,s\in\left[  0,T\right]  $ and
$h\in\left(  0,h_{0}\right]  $. As a consequence, $U_{\theta,n}^{h}\left(
t,s\right)  u$ uniformly converges to a limit $U_{\theta}^{h}\left(
t,s\right)  u$ allowing the bound (see (\ref{propagator_bound_1}))%
\begin{align}
&  \left.  \sup_{\substack{s,t\in\left[  0,T\right]  \,,\\h\in\left(
0,h_{0}\right]  }}\left\Vert U_{\theta}^{h}\left(  t,s\right)  u\right\Vert
_{L^{2}\left(  \mathbb{R}\right)  }\leq\sup\limits_{\substack{t,s\in\left[
0,T\right]  \,,\\n\in\mathbb{N}^{\ast}\,,\ h\in\left(  0,h_{0}\right]
}}\left\Vert U_{\theta,n}^{h}\left(  t,s\right)  \right\Vert _{\mathcal{L}%
\left(  L^{2}\left(  \mathbb{R}\right)  \right)  }\left\Vert u\right\Vert
_{L^{2}\left(  \mathbb{R}\right)  }\right. \nonumber\\
& \nonumber\\
&  \left.  \leq\exp\left(  C_{a,b,c}\sup_{t\in\left[  0,T\right]  }\left\Vert
\mathcal{V}\left(  t\right)  \right\Vert _{L^{\infty}\left(  \mathbb{R}%
\right)  }\right)  \left\Vert u\right\Vert _{L^{2}\left(  \mathbb{R}\right)
}\right.  \,.
\end{align}
where, according to the Remark \ref{Remark_propagator_n}, $C_{a,b,c}>0$ is
independent from $T$. Then, $U_{\theta,n}^{h}\left(  t,s\right)  $ uniformly
converges to $U_{\theta}^{h}\left(  t,s\right)  $ in the $\mathcal{L}\left(
L^{2}\left(  \mathbb{R}\right)  \right)  $ topology with%
\begin{equation}
\sup_{\substack{s,t\in\left[  0,T\right]  \,,\\h\in\left(  0,h_{0}\right]
}}\left\Vert U_{\theta}^{h}\left(  t,s\right)  \right\Vert _{\mathcal{L}%
\left(  L^{2}\left(  \mathbb{R}\right)  \right)  }\leq\exp\left(
C_{a,b,c}\sup_{t\in\left[  0,T\right]  }\left\Vert \mathcal{V}\left(
t\right)  \right\Vert _{L^{\infty}\left(  \mathbb{R}\right)  }\right)  \,.
\end{equation}

When $u\in D\left(  \Delta_{\theta}\right)  $, the sequence $U_{\theta,n}%
^{h}\left(  t,s\right)  u$ is uniformly bounded in $D\left(  \Delta_{\theta
}\right)  $ (see (\ref{propagator_bound_reg_1})). Being $D\left(
\Delta_{\theta}\right)  $ reflexive (as a subspace of $H^{2}\left(
\mathbb{R}\backslash\left\{  a,b\right\}  \right)  $), there exists a
subsequence of $U_{\theta,n}^{h}\left(  t,s\right)  u$ weakly convergent in
$D\left(  \Delta_{\theta}\right)  $. Hence: $U_{\theta}^{h}\left(  t,s\right)
u\in D\left(  \Delta_{\theta}\right)  $, and the bound%
\begin{equation}
\sup_{\substack{s,t\in\left[  0,T\right]  \,,\\h\in\left(  0,h_{0}\right]
}}\left\Vert U_{\theta}^{h}\left(  t,s\right)  u\right\Vert _{\mathcal{L}%
\left(  D\left(  \Delta_{\theta}\right)  \right)  }\leq\exp\left(  \tilde
{C}_{a,b,c}\sup_{t,s\in\left[  0,T\right]  }\left\Vert \mathcal{V}\left(
t\right)  -\mathcal{V}\left(  s\right)  \right\Vert _{W^{2,\infty}\left(
\left[  a,b\right]  \right)  }\right)  \left\Vert u\right\Vert _{D\left(
\Delta_{\theta}\right)  }\,,
\end{equation}
follows from (\ref{propagator_bound_reg_1}) with $\tilde{C}_{a,b,c}$ are
independent from $T$.
\end{proof}

For $\tilde{t}\in\left[  0,T\right]  $ fixed, let us denote with $\tilde
{U}_{\theta,n}^{h}$ the semigroup sequence obtained following the construction
(\ref{Propagator_n}) in the trivial case of the time independent Hamiltonian
$\tilde{H}_{\theta}^{h}=H_{\theta}^{h}\left(  \tilde{t}\right)  $. It results%
\begin{equation}
\tilde{U}_{\theta,n}^{h}\left(  t,s\right)  =e^{-i\left(  t-s\right)
H_{\theta}^{h}\left(  \tilde{t}\right)  }\,,\quad\forall\,n\in\mathbb{N}%
^{\ast}\,,\text{ and }t\geq s\,, \label{propagator_n_cost}%
\end{equation}
while using (\ref{Schrodinger_n}) for $\tilde{U}_{\theta,n}^{h}\left(
t,s\right)  $ and $U_{\theta,n}^{h}\left(  t,s\right)  $, we obtain the
identity%
\begin{equation}
\left(  \tilde{U}_{\theta,n}^{h}\left(  t,s\right)  -U_{\theta,n}^{h}\left(
t,s\right)  \right)  u=-i\int_{s}^{t}\tilde{U}_{\theta,n}^{h}\left(
t,t^{\prime}\right)  \left(  \mathcal{V}\left(  \tilde{t}\right)
-\mathcal{V}\left(  \frac{T}{n}\left[  \frac{nt^{\prime}}{T}\right]  \right)
\right)  U_{\theta,n}^{h}\left(  t^{\prime},s\right)  u\,dt^{\prime}\,,
\end{equation}
holding for all $u\in L^{2}\left(  \mathbb{R}\right)  $ (due the density of
the inclusion $D\left(  \Delta_{\theta}\right)  \subset L^{2}\left(
\mathbb{R}\right)  $). Taking the limit as $n\rightarrow\infty$ and using
(\ref{propagator_n_cost}), the result of the Proposition
\ref{Proposition_propagator_lim} yields%
\begin{equation}
\left(  e^{-i\left(  t-s\right)  H_{\theta}^{h}\left(  \tilde{t}\right)
}-U_{\theta}^{h}\left(  t,s\right)  \right)  u=-i\int_{s}^{t}e^{-i\left(
t-t^{\prime}\right)  H_{\theta}^{h}\left(  \tilde{t}\right)  }\left(
\mathcal{V}\left(  \tilde{t}\right)  -\mathcal{V}\left(  t^{\prime}\right)
\right)  U_{\theta}^{h}\left(  t^{\prime},s\right)  u\,dt^{\prime}\,.
\label{Cauchy_id_lim}%
\end{equation}

\begin{theorem}
\label{Theorem_propagator_lim}Under the assumptions of the Lemma
\ref{Lemma_propagator_n}, there exists an unique family of operators
$U_{\theta}^{h}\left(  t,s\right)  $, strongly continuous in $t$ and $s$
w.r.t. the $\mathcal{L}\left(  L^{2}\left(  \mathbb{R}\right)  \right)  $
topology, fulfilling the identities%
\begin{equation}
U_{\theta}^{h}\left(  s,s\right)  =1_{L^{2}\left(  \mathbb{R}\right)
}\,,\quad U_{\theta}^{h}\left(  t,s\right)  =U_{\theta}^{h}\left(  t,r\right)
U_{\theta,n}^{h}\left(  r,s\right)  \,,\quad\forall\,s\leq r\leq t\,,
\end{equation}
and such that $U_{\theta}^{h}\left(  t,s\right)  u$ is the solution of the
problem (\ref{Time_dependent_eq}) for all $u\in D\left(  \Delta_{\theta
}\right)  $.
\end{theorem}

\begin{proof}
For each $n$, $U_{\theta,n}^{h}\left(  t,s\right)  $ is strongly continuous in
$t$ and $s$ w.r.t. the $\mathcal{L}\left(  L^{2}\left(  \mathbb{R}\right)
\right)  $ topology, and fulfills the identities (\ref{propagator_id}). The
uniform convergence of the sequence in $\mathcal{L}\left(  L^{2}\left(
\mathbb{R}\right)  \right)  $ allows to extend this characterization to its
limit $U_{\theta}^{h}\left(  t,s\right)  $.

Let $u\in D\left(  \Delta_{\theta}\right)  $, $\tilde{t}\in\left[  0,T\right]
$ and $\delta>0$; the relation (\ref{Cauchy_id_lim}) yields%
\begin{equation}
U_{\theta}^{h}\left(  t,s\right)  u=e^{-i\left(  t-s\right)  H_{\theta}%
^{h}\left(  \tilde{t}\right)  }+i\int_{s}^{t}e^{-i\left(  t-t^{\prime}\right)
H_{\theta}^{h}\left(  \tilde{t}\right)  }\left(  \mathcal{V}\left(  \tilde
{t}\right)  -\mathcal{V}\left(  t^{\prime}\right)  \right)  U_{\theta}%
^{h}\left(  t^{\prime},s\right)  u\,dt^{\prime}\,.
\end{equation}
It follows%
\begin{align}
&  \left.  \left(  U_{\theta}^{h}\left(  t+\delta,s\right)  -U_{\theta}%
^{h}\left(  t,s\right)  \right)  u=\left(  e^{-i\left(  t+\delta-s\right)
H_{\theta}^{h}\left(  \tilde{t}\right)  }-e^{-i\left(  t-s\right)  H_{\theta
}^{h}\left(  \tilde{t}\right)  }\right)  u\right. \nonumber\\
& \nonumber\\
&  \left.  +\ i\left(  e^{-i\delta H_{\theta}^{h}\left(  \tilde{t}\right)
}-1\right)  \int_{s}^{t}e^{-i\left(  t-t^{\prime}\right)  H_{\theta}%
^{h}\left(  \tilde{t}\right)  }\left(  \mathcal{V}\left(  \tilde{t}\right)
-\mathcal{V}\left(  t^{\prime}\right)  \right)  U_{\theta}^{h}\left(
t^{\prime},s\right)  u\,dt^{\prime}\right. \nonumber\\
& \nonumber\\
&  \left.  +\ ie^{-i\delta H_{\theta}^{h}\left(  \tilde{t}\right)  }\int%
_{t}^{t+\delta}e^{-i\left(  t-t^{\prime}\right)  H_{\theta}^{h}\left(
\tilde{t}\right)  }\left(  \mathcal{V}\left(  \tilde{t}\right)  -\mathcal{V}%
\left(  t^{\prime}\right)  \right)  U_{\theta}^{h}\left(  t^{\prime},s\right)
u\,dt^{\prime}\right.  \,.
\end{align}
with%
\begin{equation}
\left(  e^{-i\left(  t-s\right)  H_{\theta}^{h}\left(  \tilde{t}\right)
}-U_{\theta}^{h}\left(  t,s\right)  \right)  u=-i\int_{s}^{t}e^{-i\left(
t-t^{\prime}\right)  H_{\theta}^{h}\left(  \tilde{t}\right)  }\left(
\mathcal{V}\left(  \tilde{t}\right)  -\mathcal{V}\left(  t^{\prime}\right)
\right)  U_{\theta}^{h}\left(  t^{\prime},s\right)  u\,dt^{\prime}\in D\left(
\Delta_{\theta}\right)  \,.
\end{equation}
Since $d/dt\,e^{-itH_{\theta}^{h}\left(  \tilde{t}\right)  }u=-iH_{\theta}%
^{h}\left(  \tilde{t}\right)  e^{-itH_{\theta}^{h}\left(  \tilde{t}\right)
}u$ for all $u\in D\left(  \Delta_{\theta}\right)  $, we get%
\begin{align}
&  \left.  \lim_{\delta\rightarrow0^{+}}1/\delta\left(  U_{\theta}^{h}\left(
t+\delta,s\right)  -U_{\theta}^{h}\left(  t,s\right)  \right)  u=-iH_{\theta
}^{h}\left(  \tilde{t}\right)  e^{-i\left(  t-s\right)  H_{\theta}^{h}\left(
\tilde{t}\right)  }u\right. \nonumber\\
& \nonumber\\
&  \left.  +\ H_{\theta}^{h}\left(  \tilde{t}\right)  \int_{s}^{t}e^{-i\left(
t-t^{\prime}\right)  H_{\theta}^{h}\left(  \tilde{t}\right)  }\left(
\mathcal{V}\left(  \tilde{t}\right)  -\mathcal{V}\left(  t^{\prime}\right)
\right)  U_{\theta}^{h}\left(  t^{\prime},s\right)  u\,dt^{\prime}\right.
\nonumber\\
& \nonumber\\
&  \left.  +\ \lim_{\delta\rightarrow0^{+}}i/\delta\,e^{-i\delta H_{\theta
}^{h}\left(  \tilde{t}\right)  }\int_{t}^{t+\delta}e^{-i\left(  t-t^{\prime
}\right)  H_{\theta}^{h}\left(  \tilde{t}\right)  }\left(  \mathcal{V}\left(
\tilde{t}\right)  -\mathcal{V}\left(  t^{\prime}\right)  \right)  U_{\theta
}^{h}\left(  t^{\prime},s\right)  u\,dt^{\prime}\right.  \,,
\end{align}
which leads to%
\begin{align}
&  \left.  \lim_{\delta\rightarrow0^{+}}1/\delta\left(  U_{\theta}^{h}\left(
t+\delta,s\right)  -U_{\theta}^{h}\left(  t,s\right)  \right)  u=-iH_{\theta
}^{h}\left(  \tilde{t}\right)  U_{\theta}^{h}\left(  t,s\right)  u\right.
\nonumber\\
& \nonumber\\
&  \left.  +\ \lim_{\delta\rightarrow0^{+}}i/\delta\,e^{-i\delta H_{\theta
}^{h}\left(  \tilde{t}\right)  }\int_{t}^{t+\delta}e^{-i\left(  t-t^{\prime
}\right)  H_{\theta}^{h}\left(  \tilde{t}\right)  }\left(  \mathcal{V}\left(
\tilde{t}\right)  -\mathcal{V}\left(  t^{\prime}\right)  \right)  U_{\theta
}^{h}\left(  t^{\prime},s\right)  u\,dt^{\prime}\right.  \,.
\end{align}
In particular, choosing $\tilde{t}=t$, we have%
\[
\left\Vert \int_{t}^{t+\delta}e^{-i\left(  t-t^{\prime}\right)  H_{\theta}%
^{h}\left(  \tilde{t}\right)  }\left(  \mathcal{V}\left(  \tilde{t}\right)
-\mathcal{V}\left(  t^{\prime}\right)  \right)  U_{\theta}^{h}\left(
t^{\prime},s\right)  u\,dt^{\prime}\right\Vert _{L^{2}\left(  \mathbb{R}%
\right)  }\leq\delta C_{a,b,c}\sup_{t^{\prime}\in\left[  t+\delta,t\right]
}\left\Vert \mathcal{V}\left(  t\right)  -\mathcal{V}\left(  t^{\prime
}\right)  \right\Vert =o\left(  \delta\right)  \,,
\]
and the previous limit reduces to%
\begin{equation}
\lim_{\delta\rightarrow0^{+}}1/\delta\left(  U_{\theta}^{h}\left(
t+\delta,s\right)  -U_{\theta}^{h}\left(  t,s\right)  \right)  u=-iH_{\theta
}^{h}\left(  t\right)  U_{\theta}^{h}\left(  t,s\right)  u\,.
\end{equation}
This shows that: $D_{t}^{+}U_{\theta}^{h}\left(  t,s\right)  u=-iH_{\theta
}^{h}\left(  \tilde{t}\right)  U_{\theta}^{h}\left(  t,s\right)  u$ for any
$u\in D\left(  \Delta_{\theta}\right)  $. Following the same line for the left
derivative, we obtain%
\begin{equation}
\frac{d}{dt}U_{\theta}^{h}\left(  t,s\right)  u=-iH_{\theta}^{h}\left(
t\right)  U_{\theta}^{h}\left(  t,s\right)  u\,.
\end{equation}

Assume that $V_{\theta}^{h}\left(  t,s\right)  $ is a solution of
(\ref{Time_dependent_eq}); then it expresses as%
\begin{equation}
\left(  V_{\theta}^{h}\left(  t,s\right)  -U_{\theta,n}^{h}\left(  t,s\right)
\right)  u=-i\int_{s}^{t}U_{\theta,n}^{h}\left(  t,t^{\prime}\right)  \left(
H_{\theta}^{h}\left(  t^{\prime}\right)  -H_{\theta,n}^{h}\left(  t^{\prime
}\right)  \right)  V_{\theta}^{h}\left(  t^{\prime},s\right)  u\,dt^{\prime
}\,.
\end{equation}
Estimating the difference at the r.h.s. as in (\ref{Cauchy_cond_1}%
)-(\ref{Cauchy_cond_2}), we get the identity
\begin{equation}
V_{\theta}^{h}\left(  t,s\right)  =\lim_{n\rightarrow\infty}U_{\theta,n}%
^{h}\left(  t,s\right)  =U_{\theta}^{h}\left(  t,s\right)  \,,
\end{equation}
both in the $\mathcal{L}\left(  L^{2}\left(  \mathbb{R}\right)  \right)  $ and
in the $\mathcal{L}\left(  D\left(  \Delta_{\theta}\right)  \right)  $-norm
sense. This yields the uniqueness of the solution.
\end{proof}

We are now in the position to conclude the proof of the Theorem
\ref{Theorem_main}.

\begin{proof}
[Proof of the Theorem \ref{Theorem_main}]The first part of the statement
follows from the results in the Theorem \ref{Theorem_propagator_lim}. Let us
discuss the estimate (\ref{propagator_est_global_h}). For each $n$,
$U_{\theta,n}^{h}\left(  t,s\right)  $ is $\theta$-holomorphic in
$\mathcal{L}\left(  L^{2}\left(  \mathbb{R}\right)  \right)  $ and the uniform
convergence of the sequence implies that $U_{\theta}^{h}\left(  t,s\right)  $
is $\theta$-holomorphic w.r.t. the $L^{2}\left(  \mathbb{R}\right)  $-operator
norm in the ball: $\left\vert \theta\right\vert \leq h^{N_{0}}$, $N_{0}>2$.
Let us fix $t,s$ and $h$; it results%
\begin{equation}
U_{\theta}^{h}\left(  t,s\right)  -U_{0}^{h}\left(  t,s\right)  =\theta
\,D_{\theta}^{h}\left(  t,s\right)  \,, \label{propagator_exp}%
\end{equation}
where the operator $D_{\theta}^{h}\left(  t,s\right)  $ fulfills the estimate%
\begin{equation}
\sup_{\left\vert \theta\right\vert \leq h^{N_{0}}}\left\Vert D_{\theta}%
^{h}\left(  t,s\right)  \right\Vert _{\mathcal{L}\left(  L^{2}\left(
\mathbb{R}\right)  \right)  }\leq m^{h}\left(  t,s\right)  \,,
\label{propagator_exp_1}%
\end{equation}
with $m^{h}\left(  t,s\right)  >0$ possibly depending on $h$ and on the couple
$t,s$. Fixing $\theta=h^{2+\delta}$ with $\delta>0$ arbitrarily small, it
follows%
\begin{equation}
\left\Vert U_{\theta}^{h}\left(  t,s\right)  -U_{0}^{h}\left(  t,s\right)
\right\Vert _{\mathcal{L}\left(  L^{2}\left(  \mathbb{R}\right)  \right)
}\leq h^{2+\delta}\,m^{h}\left(  t,s\right)  \,,
\label{propagator_est_global_h_1}%
\end{equation}
The uniform bound (\ref{propagator_bound_lim}) implies%
\begin{equation}
\sup_{\substack{t,s\in\left[  0,T\right]  \\h\in\left(  0,h_{0}\right]
}}\left\Vert U_{\theta}^{h}\left(  t,s\right)  -U_{0}^{h}\left(  t,s\right)
\right\Vert _{\mathcal{L}\left(  L^{2}\left(  \mathbb{R}\right)  \right)
}\leq2\exp\left(  C_{a,b,c}\sup_{t\in\left[  0,T\right]  }\left\Vert
\mathcal{V}\left(  t\right)  \right\Vert _{L^{\infty}\left(  \mathbb{R}%
\right)  }\right)  \,,\qquad t\geq s\,,
\end{equation}
Hence, (\ref{propagator_est_global_h_1}) yields%
\begin{equation}
\sup_{t,s\in\left[  0,T\right]  }m^{h}\left(  t,s\right)  \leq\frac{M_{a,b,c}%
}{h^{2+\delta}}\sup_{t\in\left[  0,T\right]  }\left\Vert \mathcal{V}\left(
t\right)  \right\Vert _{L^{\infty}\left(  \mathbb{R}\right)  }\,.
\label{propagator_exp_lim}%
\end{equation}
with $M_{a,b,c}>0$ possibly depending on the data, but independent from $T$.
The estimate (\ref{propagator_est_global_h}) finally follows from
(\ref{propagator_exp})-(\ref{propagator_exp_1}) and (\ref{propagator_exp_lim}).
\end{proof}

\bigskip

\bigskip

\bigskip

\bigskip
\end{document}